\documentclass[10pt]{article}

\usepackage[left=1in,top=1in,right=1in,bottom=1in,letterpaper]{geometry}

\usepackage{amsmath,amssymb,amsfonts}
\usepackage{amsthm}
\usepackage{url}

\usepackage[usenames]{color}
\usepackage[color,notcite,final]{showkeys}
\definecolor{refkey}{gray}{0.8}
\definecolor{labelkey}{gray}{0.8}

       \usepackage[colorlinks=true, pdfstartview=FitV, linkcolor=blue, citecolor=blue, urlcolor=blue]{hyperref}
       \usepackage{graphicx}

\usepackage{marvosym}

\usepackage[normalem]{ulem} 

\newcommand{\va}{{\mathbf{a}}}
\newcommand{\vb}{{\mathbf{b}}}
\newcommand{\vc}{{\mathbf{c}}}
\newcommand{\vd}{{\mathbf{d}}}

\newcommand{\vh}{{\mathbf{h}}}

\newcommand{\vn}{{\mathbf{n}}}

\newcommand{\vp}{{\mathbf{p}}}

\newcommand{\vs}{{\mathbf{s}}}
\newcommand{\vt}{{\mathbf{t}}}
\newcommand{\vu}{{\mathbf{u}}}
\newcommand{\vv}{{\mathbf{v}}}

\newcommand{\vx}{{\mathbf{x}}}
\newcommand{\vy}{{\mathbf{y}}}
\newcommand{\vz}{{\mathbf{z}}}

\newcommand{\vA}{{\mathbf{A}}}
\newcommand{\vB}{{\mathbf{B}}}
\newcommand{\vC}{{\mathbf{C}}}
\newcommand{\vD}{{\mathbf{D}}}

\newcommand{\vH}{{\mathbf{H}}}

\newcommand{\vL}{{\mathbf{L}}}
\newcommand{\vM}{{\mathbf{M}}}

\newcommand{\vS}{{\mathbf{S}}}

\newcommand{\vU}{{\mathbf{U}}}
\newcommand{\vV}{{\mathbf{V}}}

\newcommand{\vX}{{\mathbf{X}}}
\newcommand{\vY}{{\mathbf{Y}}}

\newcommand{\cA}{{\mathcal{A}}}

\newcommand{\cM}{{\mathcal{M}}}

\newcommand{\cS}{{\mathcal{S}}}

\newcommand{\cV}{{\mathcal{V}}}

\newcommand{\cY}{{\mathcal{Y}}}
\newcommand{\cZ}{{\mathcal{Z}}}

\newcommand{\RR}{\mathbb{R}}

\newcommand{\sign}{\mathrm{sign}}
\newcommand{\vzero}{\mathbf{0}}
\newcommand{\vone}{{\mathbf{1}}}

\newcommand{\supp}{{\mathrm{supp}}} 
\newcommand{\diag}{{\mathrm{diag}}} 
\newcommand{\dom}{{\mathrm{dom}}} 
\newcommand{\grad}{{\nabla}}    
\newcommand{\Proj}{{\mathrm{Proj}}} 
\newcommand{\Range}{{\mathrm{Range}}} 
\DeclareMathOperator{\shrink}{shrink} 
\DeclareMathOperator*{\argmin}{arg\,min}

\DeclareMathOperator{\rank}{rank}
\DeclareMathOperator{\trace}{trace}
\DeclareMathOperator{\Null}{Null}
\DeclareMathOperator{\dist}{dist}

\newcommand{\bc}{\begin{center}}
\newcommand{\ec}{\end{center}}

\newcommand{\bdm}{\begin{displaymath}}
\newcommand{\edm}{\end{displaymath}}

\newcommand{\beq}{\begin{equation}}
\newcommand{\eeq}{\end{equation}}

\newcommand{\bfl}{\begin{flushleft}}
\newcommand{\efl}{\end{flushleft}}

\newcommand{\bt}{\begin{tabbing}}
\newcommand{\et}{\end{tabbing}}

\newcommand{\beqn}{\begin{eqnarray}}
\newcommand{\eeqn}{\end{eqnarray}}

\newcommand{\beqs}{\begin{eqnarray*}} 
\newcommand{\eeqs}{\end{eqnarray*}}  

\newtheorem{theorem}{Theorem}

\newtheorem{definition}{Definition}

\newtheorem{remark}{Remark}

\newtheorem{lemma}{Lemma}

\usepackage{subfigure}

\begin{document}

\title{Augmented $\ell_1$ and Nuclear-Norm Models  with a Globally Linearly Convergent Algorithm}
\author{Ming-Jun Lai\thanks{Department of Mathematics,
The University of Georgia, Athens, GA 30602. Email: \url{mjlai@math.uga.edu}.}
\and
Wotao Yin\thanks{Department of Applied and Computational Mathematics, Rice University, Houston, TX. This author is partly supported by NSF grants
DMS-0748839 and ECCS-1028790, ONR grant N00014-08-1-1101, and  ARL and ARO grant W911NF-09-1-0383. Email: \url{wotao.yin@rice.edu}.}
}
\date{}
\maketitle

\begin{abstract}
This paper studies the long-existing idea of adding a nice smooth function to ``smooth'' a non-differentiable objective function in the context of sparse optimization, in particular, the minimization of $\|\vx\|_1+\frac{1}{2\alpha}\|\vx\|_2^2$, where $\vx$ is a vector, as well as  the minimization of
$\|\vX\|_*+\frac{1}{2\alpha}\|\vX\|_F^2$, where $\vX$ is a matrix and $\|\vX\|_*$ and $\|\vX\|_F$ are the nuclear and Frobenius norms of $\vX$,
respectively. We show that they let sparse vectors and low-rank matrices be efficiently recovered. In particular, they enjoy exact and stable recovery
guarantees
 similar to those known for the minimization of $\|\vx\|_1$ and $\|\vX\|_*$  under  the conditions on the sensing operator such as its  null-space property,
restricted isometry property, spherical section property, or  ``RIPless'' property. To recover a (nearly) sparse vector  $\vx^0$, minimizing
$\|\vx\|_1+\frac{1}{2\alpha}\|\vx\|^2_2$ returns (nearly) the same solution as minimizing $\|\vx\|_1$  whenever $\alpha\ge 10\|\vx^0\|_\infty$.
The same  relation also holds between minimizing $\|\vX\|_*+\frac{1}{2\alpha}\|\vX\|_F^2$ and minimizing $\|\vX\|_*$  for recovering a (nearly)
low-rank matrix  $\vX^0$ if  $\alpha\ge 10\|\vX^0\|_2$. Furthermore, we show that the linearized Bregman algorithm, as well as its two fast variants, for minimizing
$\|\vx\|_1+\frac{1}{2\alpha}\|\vx\|_2^2$ subject to $\vA\vx=\vb$ enjoys \emph{global} linear convergence as long as a nonzero solution exists,  and we
give an explicit rate of convergence. The convergence property does not require a sparse solution or any properties on  $\vA$. To our knowledge,
this is the best known global convergence result for  first-order sparse optimization algorithms.
\end{abstract}

\section{Introduction}
Sparse vector recovery and low-rank matrix recovery problems have drawn lots of attention from researchers in different
fields in the past several years.
They have wide applications in compressive
sensing, signal/image processing, machine learning, etc.
The fundamental problem of sparse vector recovery is to find the
 vector with (nearly) fewest nonzero entries   from an {underdetermined} linear system $\vA\vx = \vb$, and that of low-rank matrix recovery is to find a
matrix of (nearly) lowest rank from an {underdetermined}
$\cA(\vX)=\vb$, where $\cA$ is a linear operator.

To recover a sparse vector $\vx^0$, a well-known model is the basis pursuit problem \cite{Chen-Donoho-Saunders-98}:
\begin{equation}\label{P1}
\min_\vx\{\|\vx\|_1: \vA \vx = \vb\}.
\end{equation}
For vector $\vb$ with noise or  generated by an approximately sparse vector, a variant of \eqref{P1} is
\begin{equation}\label{P1n}
\min_\vx\{\|\vx\|_1: \|\vA \vx - \vb\|_2\le\sigma\}.
\end{equation}
To recover a low-rank matrix  $\vX^0\in\mathbb{R}^{n_1\times n_2}$ from linear measurements $\vb=\cA(\vX^0)$, which stand for $b_i =\trace(\vA_i^\top
\vX^0)$ for a  given matrix $\vA_i\in\mathbb{R}^{n_1\times n_2}$, $i=1,2,\ldots, m$, a popular approach  is  the convex model
(cf. \cite{Fazel-thesis-02,CR08,Recht-Fazel-Parrilo-07})
\begin{equation}
\label{N1}
\min_{\vX}\left\{\|\vX\|_*: \cA(\vX) = \vb\right\},
\end{equation}
where  $\|\vX\|_*$ equals the summation of the singular values of $\vX$. Similar to \eqref{P1n}, a useful variant of \eqref{N1} is
\beq\label{N1n}
\min_{\vX}\left\{\|\vX\|_*: \|\cA(\vX) - \vb\|_2\le \sigma\right\},
\eeq

The nonsmooth objective functions in problems \eqref{P1}--\eqref{N1n} pose numerical challenges. We augment or ``smooth'' them by adding  $\frac{1}{2\alpha}\|\vx\|_2^2$ or $\frac{1}{2\alpha}\|\vX\|_F^2$, where $\alpha$ is a positive scalar. We argue that minimizing the \emph{augmented}  objective $\|\vx\|_1+\frac{1}{2\alpha}\|\vx\|_2^2$, as well as $\|\vX\|_*+\frac{1}{2\alpha}\|\vX\|_F^2$, leads to fast numerical algorithms  because not only accurate solutions can be obtained by using a sufficiently large, yet not excessive large, value of $\alpha$, but the Lagrange dual problems are also continuously differentiable and subject to  gradient-based acceleration techniques such as line search.

Next, we briefly review the related works and summarize the contributions of this paper.
The augmented model for \eqref{P1} is
\beq
\label{P2}
\min_\vx \left\{\|\vx\|_1 + \frac{1}{2\alpha} \|\vx\|_2^2: \vA\vx = \vb\right\},
\eeq
which can be solved by the linearized Bregman algorithm (LBreg) \cite{YOGD08}, which is  analyzed in \cite{Cai-Osher-Shen-lbreg2-09,Yin-LBreg-09}.
(Note that LBreg is different from  the Bregman algorithm \cite{Osher-Burger-Goldfarb-Xu-Yin-05,YOGD08}, which solves problem \eqref{P1} instead of  \eqref{P2}.) 


The  \emph{exact regularization property} of \eqref{P2} is proved in \cite{Yin-LBreg-09}: the solution to \eqref{P2} is also a solution to \eqref{P1} as long as $\alpha$
is sufficiently large. The property can also be obtained from \cite{Friedlander-Tseng-07}. However, neither paper tells how to select $\alpha$, whereas the size of $\alpha$ affects the numerical
performance. It has been observed by several groups of researchers that a larger $\alpha$ tends to cause slower convergence. Hence, one would like to
choose a moderate $\alpha$ that is just  large enough for \eqref{P2} to return a solution to \eqref{P1}. For recovering a sparse vector $\vx^0$ and a  low-rank matrix $\vX^0$, this paper gives the simple formulae
$$\alpha\ge 10\|\vx^0\|_\infty\quad\text{and}\quad \alpha\ge 10 \|\vX^0\|_2,$$ respectively, where the operator norm $\|\vX^0\|_2$ equals the maximum singular value of $\vX^0$. Although $\vx^0$ and $\vX^0$ are not known when $\alpha$ must be set, $\|\vx^0\|_\infty$ and $\|\vX^0\|_2$ are often easy to estimate. For example, in
compressive sensing,  $\|\vx^0\|_\infty$ is  the maximum intensity of the underlying signal or the maximum sensor reading. When the total energy $\|\vx^0\|_2$ is roughly known, one can apply the more conservative formula: $\alpha\ge 10\|\vx^0\|_2$ since  $\|\vx^0\|_2\ge \|\vx^0\|_\infty$.  Similarly, a more conservative formula is $\alpha\ge 10\|\vX^0\|_F$ for the matrix case.

This paper also shows that the Lagrange dual problem of \eqref{P2} is unconstrained and differentiable, and its objective is uniformly strongly convex when restricted to certain pairs of points. Consequently, algorithm LBreg, as well as two faster variants, enjoys \emph{global} linear convergence; specifically, both the objective error and solution error are bounded by $O(\mu^{k})$, where $k$ is the iteration number and  $\mu$ is a constant strictly less than  $1$. The value of $\mu$ depends on $\alpha$, the dynamic range of
the solution's nonzero entries, as well as some properties of $\vA$. Although several first-order algorithms for \eqref{P1} have been shown to have
asymptotic linear convergence, this is the first global linear convergence result that comes with an explicit rate.

We shall discuss strong convexity.    Many of the algorithms for recovering sparse solutions from under-determined systems of equations are observed to have a linearly converging behavior, at least on  problems that are not severely ``ill-conditioned''; however, their underlying objective functions do not have strong convexity -- a property commonly used to ensure global linear convergence -- when the linear operator $\vA$ has  fewer rows than columns. Specifically, the loss function in the form of $g(\vA\vx-\vb)$, even for strongly convex function  $g$, is ``flat'' along many directions. Flatness or near flatness along a direction means  a small directional gradient, which can generally cause slow decrease in the objective value. However, in problems with certain types of matrix $\vA$, moving along these directions will significantly change  the regularization function. In the recent paper \cite{ANW12}, the definition of strong convexity is extended to include a relaxation term involving the regularizer function. The paper argues that, with high probability for problems with $\vA$ that is random or satisfies restricted eigenvalue  or other suitable properties, their ``restricted strong convexity'' definition is satisfied by the sum of the regularization and loss functions, and as a result, the prox-linear or gradient projection iteration applied to minimizing the sum   has a (nearly-)linear convergence behavior, specifically,
$$\|\vx^{(k+1)}-\vx^*\|^2\le c^k\|\vx^{(0)}-\vx^*\|^2 + o(\|\vx^*-\vx^0\|^2),$$
where $c<1$, $\vx^*$ and $\vx^0$ are the  minimizer and underlying true signal, respectively, and $\vx^{(k)}$ stands for the $k$th iterate.
This paper presents a different approach. Due to smoothing, unmodified linear convergence to the exact solution is achievable without a probabilistic argument. The Lagrange dual of \eqref{P2}  is  strongly convex, not in the global sense, but restricted between the current point and its projection to the solution set. This property turns out to be sufficient for global linear convergence without a  modification.

Numerically, LBreg  without acceleration  is not very efficient because  it is equivalent to  the  dual gradient ascent with a \emph{fixed step size}, as shown in \cite{Yin-LBreg-09}. Nonetheless, the step size can be relaxed. Since the augmentation term $\frac{1}{2\alpha} \|\vx\|_2^2$ makes the dual problem   unconstrained and differentiable, the dual is subject to advanced gradient-descent techniques such as Barzilai-Borwein (BB) step sizes \cite{Barzilai-Borwein-88}, non-monotone line search, Nesterov's technique \cite{Nesterov-83}, as well as semi-smooth Newton methods. Indeed, LBreg has been improved in several recent  works: \cite{Osher-Mao-Dong-Yin-10} applies a kicking trick; \cite{Yin-LBreg-09} considers applying BB step sizes and non-monotone line search, as well as the limited memory BFGS method \cite{Liu-Nocedal-89};
\cite{Yang-Moller-Osher-11} applies the alternating direction method to the Lagrange dual of \eqref{P2}; \cite{Huang-Ma-Goldfarb-11} applies Nesterov's technique \cite{Nesterov-83} and obtains the convergence rate $O(1/k^2)$. Based on the restricted strong convexity of the dual objective and some existing proofs, we theoretically show and numerically demonstrate that LBreg with BB  step sizes with non-monotone line search also enjoys global linear convergence.

LBreg has also been extended to recovering simply structured matrices.  The algorithms SVT \cite{Cai-Candes-Shen-08} for matrix completion and
IT \cite{Wright-Ganesh-Rao-Ma-09} for robust principal components are of the LBreg type, namely, they are gradient iterations that solve
\begin{align}\label{mcdl}
\min_{\vX} \{\|\vX\|_*+\frac{1}{2\alpha}\|\vX\|_F^2:\vX_{ij}=\vM_{ij}, ~\forall\,(i,j)\in\Omega\},\\
\label{rpca}
\min_{\vL,\vS} \{\|\vL\|_*+\lambda\|\vS\|_1+\frac{1}{2\alpha}\left(\|\vL\|_F^2+\|\vS\|_F^2\right):\vL+\vS=\vD\},
\end{align}
respectively, where $\Omega$ is the set of the observed matrix entries and $\|\vS\|_1=\sum_{i,j}|S_{i,j}|$.
\cite{Zhang-Cheng-Zhu-11} shows that the exact regularization property for the vector case
also holds for \eqref{mcdl} and \eqref{rpca}. Although this paper does not analyze \eqref{mcdl} and \eqref{rpca} specifically, it gives recovery
guarantees for 
models
\beq\label{N2}
\min_\vx \left\{\|\vX\|_* + \frac{1}{2\alpha} \|\vX\|_F^2: \cA(\vX) = \vb\right\}
\eeq
and
\beq\label{N3}
\min_\vx \left\{\|\vX\|_* + \frac{1}{2\alpha} \|\vX\|_F^2: \|\cA(\vX) - \vb\|_2\le \sigma\right\}
\eeq
assuming $\alpha \ge 10 \|\vX^0\|_2$. 

\subsection{Organization}
The rest of this paper is organized as follows. Section \ref{sc:models} presents several models with augmented $\ell_1$ or augmented nuclear-norm objectives and
derives their Lagrange dual problems. The exact and stable recovery conditions for  these models are given  in Section \ref{rcg}. Section \ref{sc:glc}  proves a
restricted strongly convex property and  establishes global linear convergence for LBreg and its two faster variants. The materials of Sections 3 and 4 are technically
independent of each other, yet they are two important sides of model \eqref{P2}.

The matlab codes and demos of LBreg, including the original, line search, and Nesterov acceleration versions, can be found from the second author's homepage.



\section{Augmented $\ell_1$ and nuclear-norm models}\label{sc:models}


This section presents the primal and dual problems of a few augmented $\ell_1$ and augmented nuclear-norm models.

\textbf{Equality constrained augmented $\ell_1$ model}: Since $\|\vx\|_1= \max \{\vx^\top \vz: \vz\in \mathbb{R}^n, \|\vz\|_\infty\le 1\}$,
the dual problem of \eqref{P2} can be obtained as follows
\begin{align*}
\min_\vx\{\|\vx\|_1+ \frac{1}{2\alpha} \|\vx\|_2^2: \vA \vx = \vb\}=& \min_{\vx}\max_{\vy} \|\vx\|_1 +
\frac{1}{2\alpha} \|\vx\|_2^2 - \vy^\top (\vA\vx -\vb)\cr
=& \min_{\vx } \max_{\vy,\vz} \{\vx^\top \vz +\frac{1}{2\alpha}\|\vx\|_2^2 - \vy^\top \vA \vx + \vy^\top \vb:\|\vz\|_\infty\le 1\} \cr
=&\max_{\vy,\vz} \{\min_{\vx }     \vx^\top \vz + \frac{1}{2\alpha}\|\vx\|_2^2- \vy^\top \vA \vx + \vb^\top \vy:\|\vz\|_\infty\le 1\} \cr
=& -\min_{\vy,\vz} \{ -\vb^\top \vy + \frac{\alpha}{2}\|\vA^\top \vy - \vz\|_2^2 : \|\vz\|_\infty\le 1 \}, \quad \hbox{since  } \vx^* =
\alpha(\vA^\top\vy - \vz).
\end{align*}
Eliminating $\vz$ from the last equation gives the following dual problem.
\begin{equation}\label{D2u}
\min_{\vy}\, -\vb^\top \vy + \frac{\alpha}{2} \|\vA^\top \vy - \Proj_{[-1,1]^n}(\vA^\top \vy)\|_2^2.
\end{equation}
For any real vector $\vz$, we have $\vz-\Proj_{[-\mu,\mu]^n}(\vz)=\shrink_\mu(\vz)$, where $\shrink_\mu$ is the well-known shrinkage or
soft-thresholding operator with parameter $\mu>0$. We omit $\mu$ when $\mu=1$. Hence, the second
term in \eqref{D2u} equals $(\alpha/2)\|\shrink(\vA^\top\vy)\|_2^2$.

It is interesting to compare \eqref{D2u} with the Lagrange dual of \eqref{P1}:
\beq\label{P1d}
\min_{\vy}\{ -\vb^\top \vy: \|\vA^\top \vy\|_\infty\le 1\}.
\eeq
Instead of confining each component of $\vA^\top\vy$ to $[-1,1]$, \eqref{D2u} applies quadratic penalty to the violation. This leads to its advantage
of being unconstrained and differentiable (despite the presence of projection).

The gradient of the last term in \eqref{D2u} is $\alpha\vA\shrink(\vA^\top\vy)$. Furthermore, given a solution $\vy^*$ to \eqref{D2u}, one can
recover the
solution $\vx^*=\alpha\shrink(\vA^\top \vy^*)$ to \eqref{P2} (since  \eqref{D2u} has a vanishing gradient  $\vA\vx^*-\vb=\vzero$, and $\vx^*$ and
$\vy^*$ lead to  0-gap primal and dual objectives, respectively).
Therefore, solving \eqref{D2u} solves \eqref{P2}, and it is easier than solving \eqref{P1}. In particular, \eqref{D2u} enjoys a rich set of classical techniques such as line search, Barzilai-Borwein steps \cite{Barzilai-Borwein-88}, semi-smooth Newton methods, Nesterov's acceleration \cite{Nesterov-83}, which do not directly apply to problems \eqref{P1} or \eqref{P1d}.

\textbf{Norm-constrained augmented $\ell_1$}: 
For model \eqref{P1n}, the primal
and dual  augmented models are
\begin{align}
\label{P3}
& \min_\vx \left\{\|\vx\|_1 + \frac{1}{2\alpha} \|\vx\|_2^2: \|\vA\vx - \vb\|_2\le \sigma\right\},\\
\label{D3}
& \min_{\vy} \left\{-\vb^\top \vy +\sigma\|\vy\|_2 + \frac{\alpha}{2} \|\vA^\top \vy - \Proj_{[-1,1]^n}(\vA^\top \vy)\|_2^2\right\}.
\end{align}
The objective of \eqref{D3}
is differentiable except at  $\vy=\vzero$. However, this is not an issue since $\vy=\vzero$ is a solution to \eqref{D3} only if $\vx=\vzero$ is the solution to \eqref{P3}. 
In other words, \eqref{D3} is practically differentiable and thus also amenable to classical gradient-based acceleration techniques.

\textbf{Equality-constrained augmented $\|\cdot\|_*$}: 
The primal and dual  of the augmented model of \eqref{N1} are \eqref{N2} and
\beq
\label{M2}
\min_{\vy}\left\{ -\vb^\top \vy + \frac{\alpha}{2} \|\cA^* \vy - \Proj_{\{\vX:\|\vX\|_2\le 1\}}(\cA^* \vy)\|_F^2\right\},
\eeq
respectively, where $\cA^* \vy := \sum_{i=1}^m y_i \vA_i$ and $\{\vX:\|\vX\|_2\le 1\}$ is the set of $n_1$-by-$n_2$ matrices with spectral norms no more than 1. In \eqref{M2}, inside the Frobenius norm is the singular value soft-thresholding \cite{Cai-Candes-Shen-08} of $\cA^*\vy$.

The primal and dual of the  augmented model \eqref{N1n} are \eqref{N3} and
\beq
\label{M3}
\min_{\vy} \left\{-\vb^\top \vy +\sigma\|\vy\|_2 + \frac{\alpha}{2} \|\cA^* \vy - \Proj_{\{\vX:\|\vX\|_2\le 1\}}(\cA^* \vy)\|_F^2\right\},
\eeq
respectively. Like the augmented models for vectors, problems \eqref{M2} and \eqref{M3} are practically differentiable
and thus
also amenable to advanced optimization techniques for unconstrained differentiable problems.

As one can see, it is a routine task to augment an $\ell_1$-like minimization problem and obtain a problem with a strongly convex objective, as well as  its  Lagrange dual with a differentiable objective and no constraints. One can augment models with a transform-$\ell_1$ objective, total variation, $\ell_{1,2}$ or $\ell_{1,\infty}$ norms (for joint or group sparse signal recovery), robust-PCA objective, etc. Since the dual problems are convex and differentiable,  they enjoy a rich set of gradient-based optimization  techniques.

\section{Recovery Guarantees}\label{rcg}
This section establishes recovery guarantees for augmented $\ell_1$ models \eqref{P2} and \eqref{P3} and extend these results to matrix recovery models \eqref{N2} and \eqref{N3}.  The results for \eqref{P2} and \eqref{P3} are given based on  a variety of properties of $\vA$ including 
the null-space property (NSP) in Theorem \ref{thm:nsp}, the restricted isometry property (RIP) \cite{Candes-Tao-05} in Theorems \ref{main1} and \ref{main2}, the spherical section property (SSP) \cite{Zhang-08} in Theorems \ref{thm:ssp} and \ref{thm:ssp1}, and an ``RIPless'' condition \cite{Candes-Plan-10} in Theorem \ref{thm:ripless} below. 
We choose to study all these different properties since  they give different types of recovery guarantees and  apply to different type of matrices. Other than that NSP is used in our proofs for RIP and SSP, the other three properties --- RIP, SSP, and RIPless --- do not dominate one another in terms of  usefulness.  They together assert that a large number of matrices such as those sampled from subgaussian distributions,  Fourier and Wash-Hadamard ensembles, and random Toeplitz and circulant ensembles are suitable for sparse vector recovery by models \eqref{P2} and \eqref{P3}. 

First, we present some numerical simulations to motivate the subsequent analysis.

\subsection{Motivating examples}\label{motv}
We are interested in  comparing model \eqref{P2} to model \eqref{P1}, whose  the  performance on recovering sparse solutions have been widely studied. To this end, we conducted three sets of simulations. 
Without loss of generality, we fixed $\|\vx^0\|_\infty=1$ and solved \eqref{P1} and then \eqref{P2} with $\alpha = 1, 10$, and $25$ to reconstruct signals
of $n=400$ dimensions. We set the signal sparsity   $k=1,2,\ldots, 80$ and the number of measurements  $m = 40, 41,\ldots, 200$. The entries of $\vA$
were sampled from the standard Gaussian distribution.

It turns out that the recovery performance of
\eqref{P2} depends on the decay speed of the  nonzero entries of the signal $\vx^0$. So, we tested three  decay speeds: (i) flat magnitude --- no
decay, (ii) independent Gaussian --- moderate decay, and (iii) power law --- fast decay. In the power-law decay, the $i$th largest entry had magnitude $
i^{-2}$ and a random sign. 

For each $(m,k)$,  100 independent tests were run, and the average of \beq\label{relerr}
\text{recovery relative error}\quad \|\vx^* - \vx^0\|_2/\|\vx^0\|_2
\eeq
was recorded, where $\vx^*$ stands for a solution of either \eqref{P1} or \eqref{P2}. The slightly smoothed cut-off curves at two different levels of relative errors are depicted in Figure \ref{fig1}. \emph{Above} each curve is the region where a model \emph{fails} to recover the signals to the specified average
relative error. Hence, a \emph{higher}  curve means  fewer  fails and thus \emph{better} recovery performance.
\begin{figure}[h]
\begin{center}
\subfigure[\textbf{Flat} and $10^{-3}$]{
\includegraphics[width=0.31\textwidth]{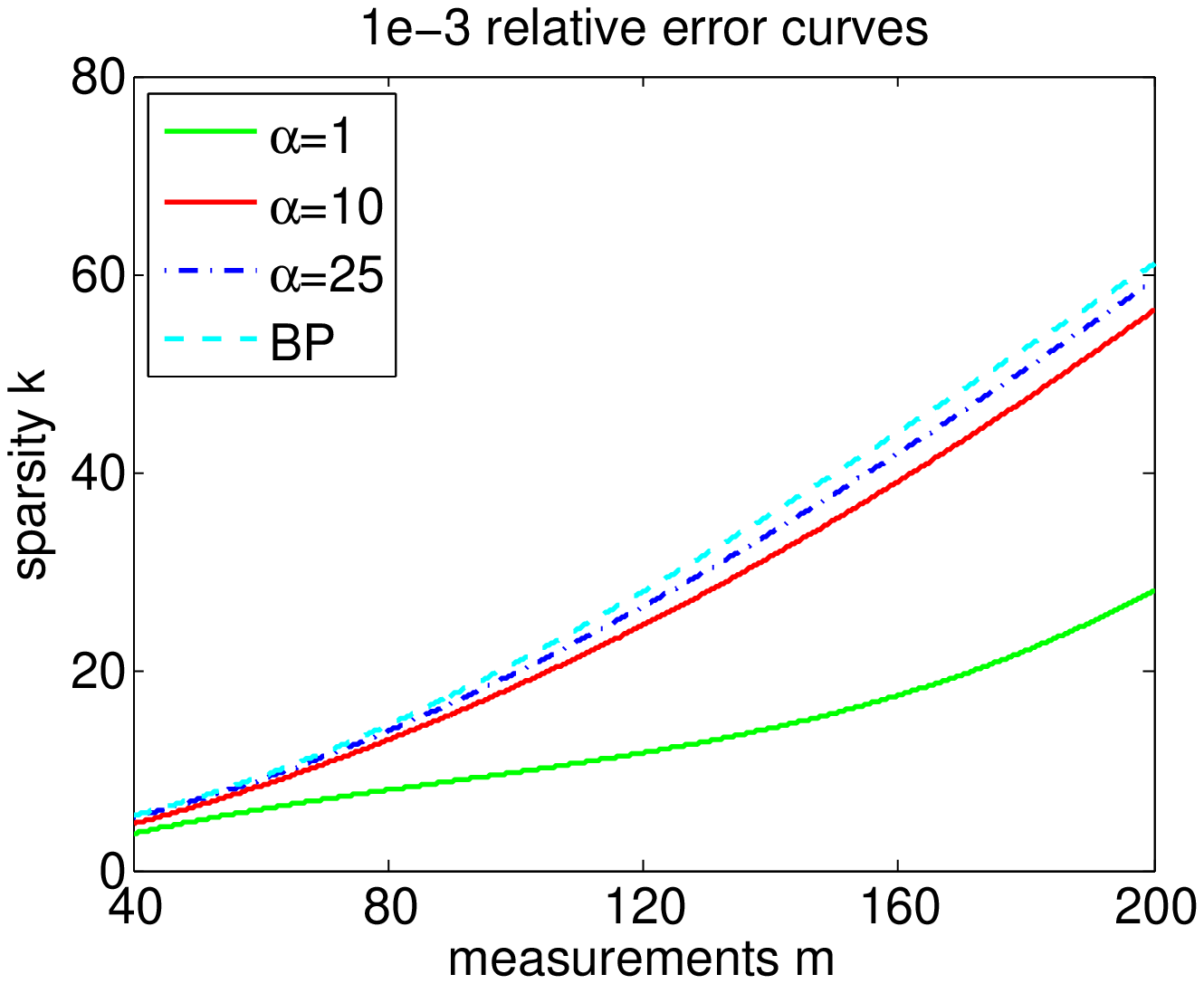}
}
\subfigure[\textbf{Gaussian}  and $10^{-3}$]{
\includegraphics[width=0.31\textwidth]{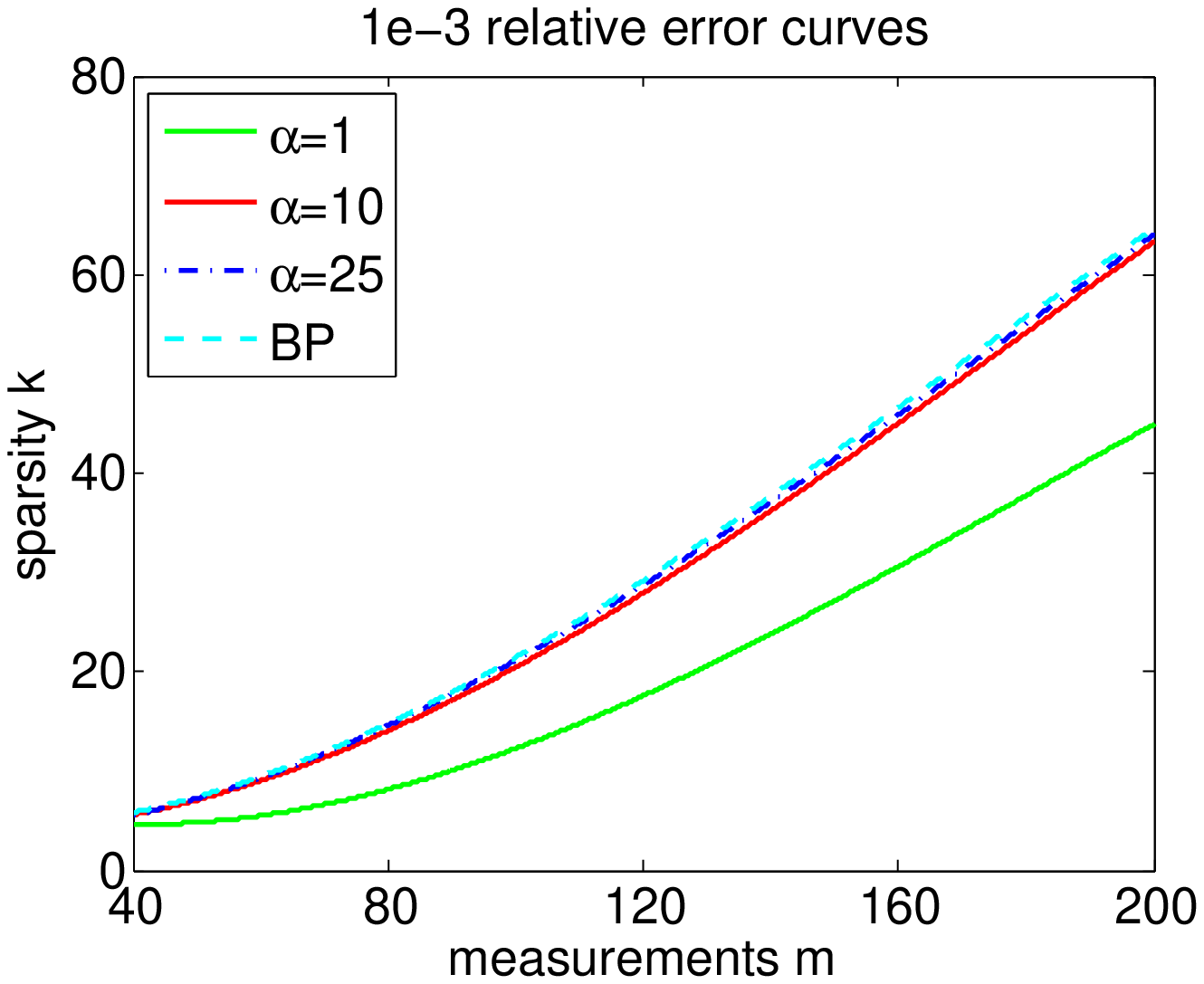}
}
\subfigure[\textbf{Power-law} and $10^{-3}$]{
\includegraphics[width=0.31\textwidth]{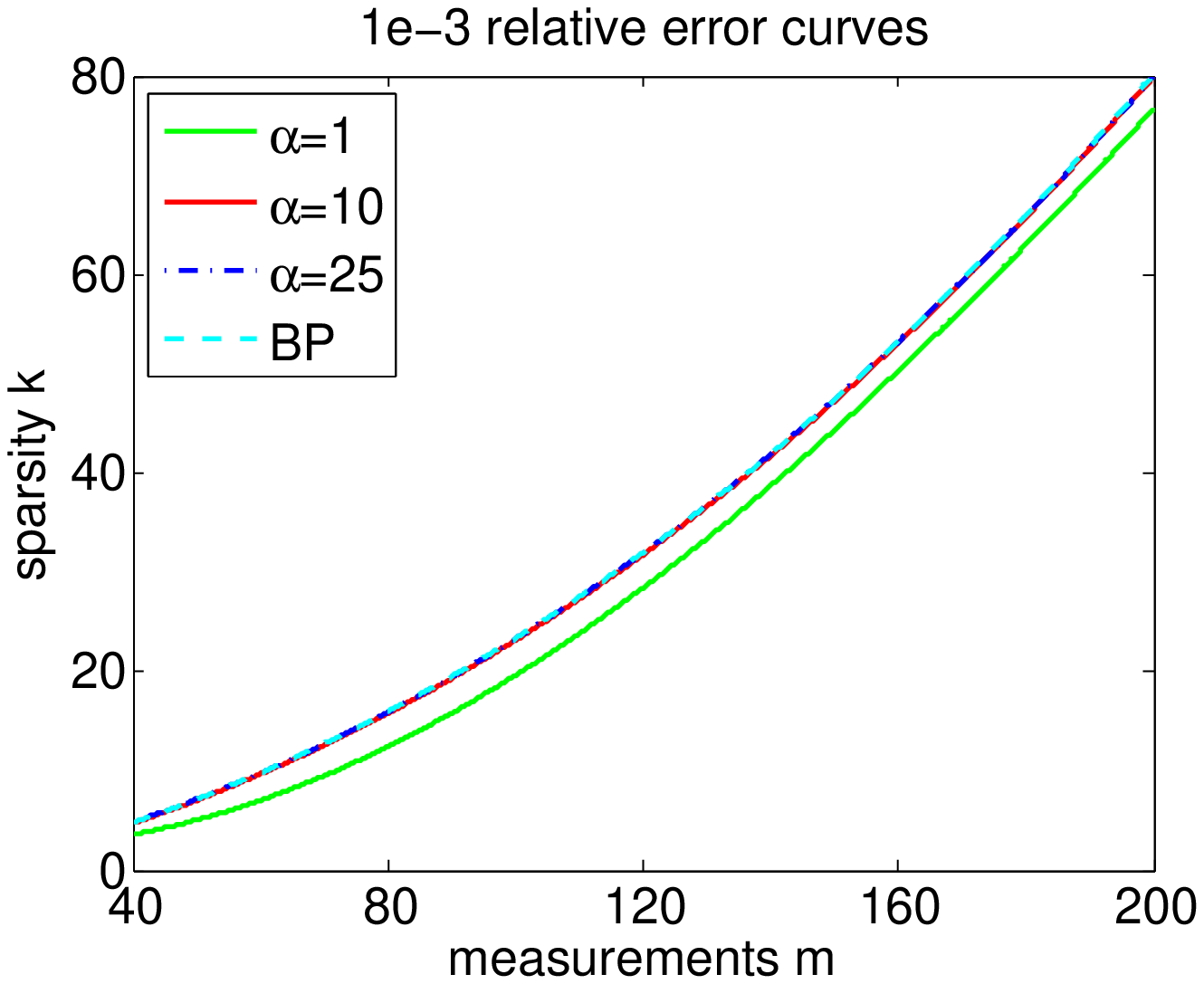}
}\\
\subfigure[\textbf{Flat} and $10^{-5}$]{
\includegraphics[width=0.31\textwidth]{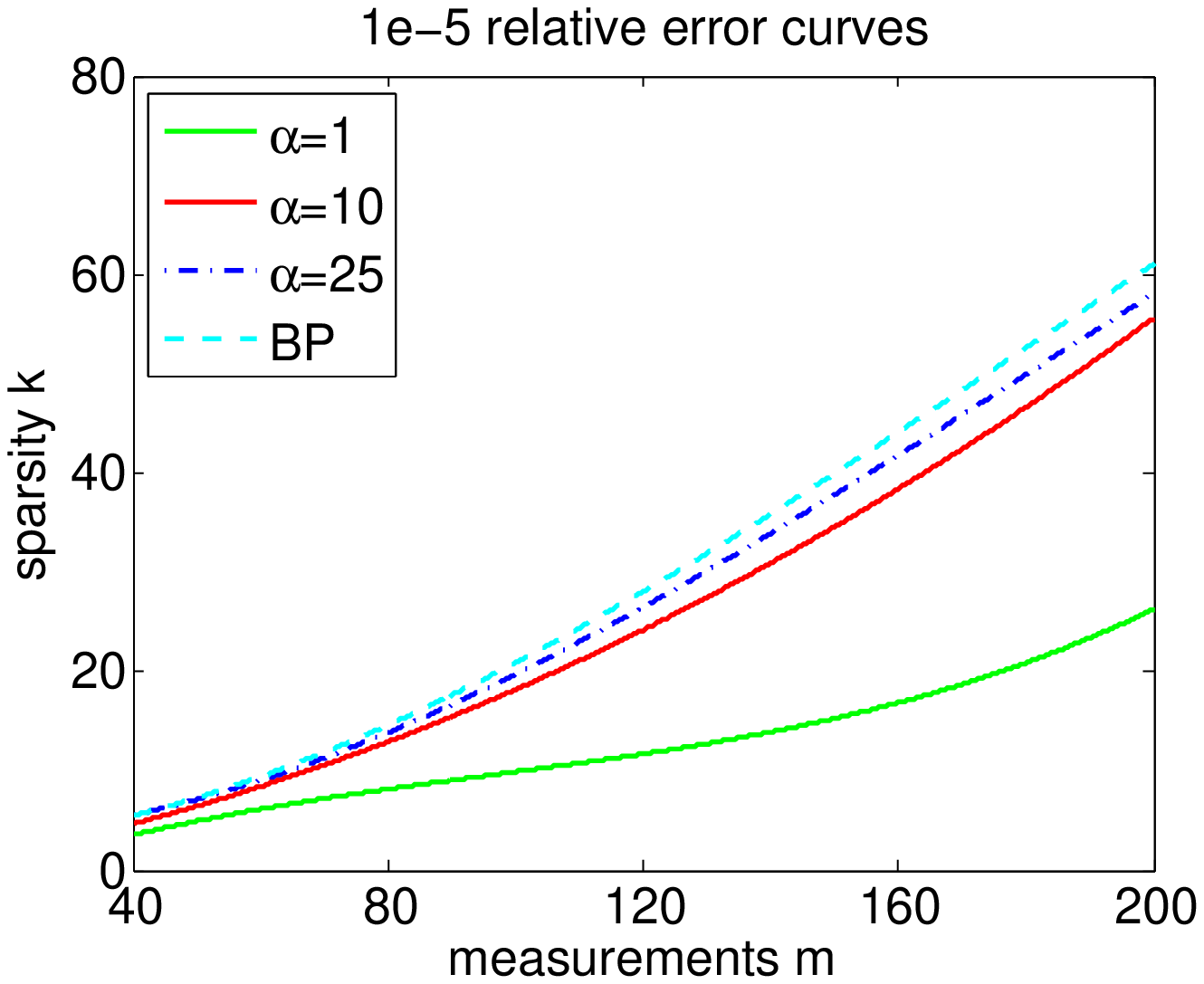}
}
\subfigure[\textbf{Gaussian} and $10^{-5}$]{
\includegraphics[width=0.31\textwidth]{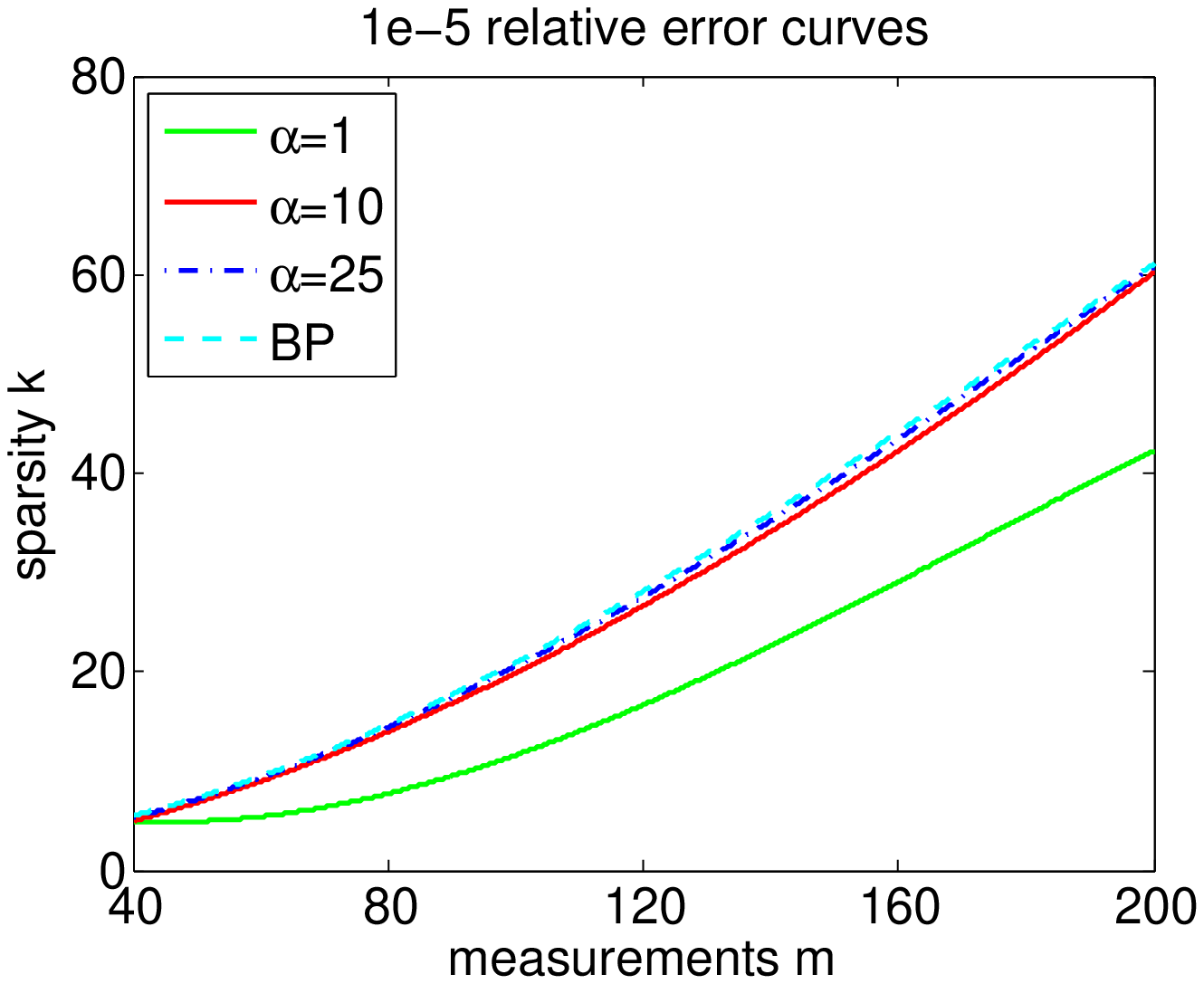}
}
\subfigure[\textbf{Power-law} and $10^{-5}$]{
\includegraphics[width=0.31\textwidth]{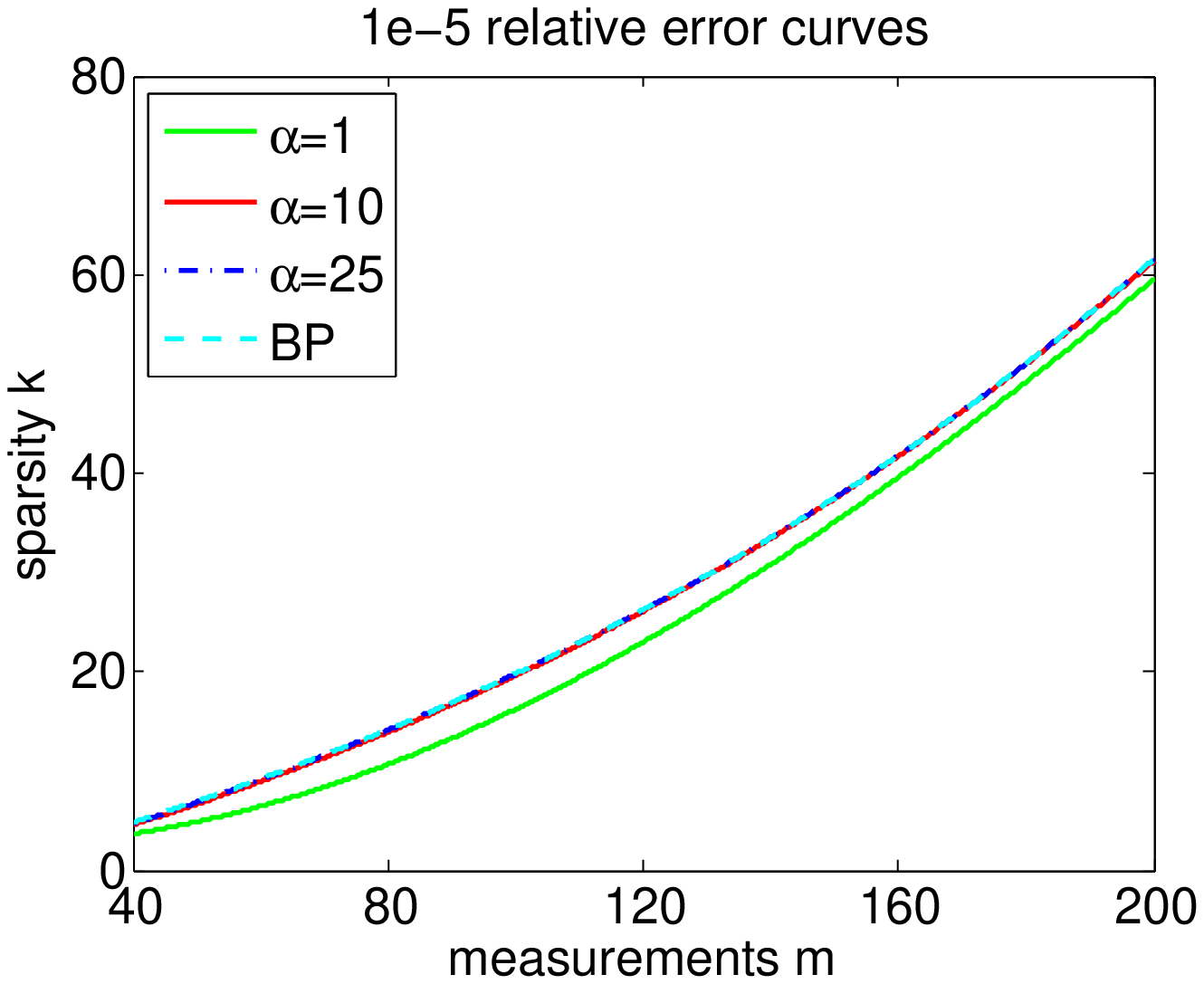}
}
\end{center}
\caption{\textbf{Curves of  specified recovery relative errors} of model \eqref{P1} (BP) and model \eqref{P2} with $\alpha=1,10,25$. Above each curve is the region where a model \emph{fails} to achieve the specified average
relative error. A higher curve means better recovery performance.\label{fig1}}
\end{figure}

We can make following observations.

\begin{itemize}
\item In all tests, the best curve is from BP or model \eqref{P1}. Closely following it are those of $\alpha=25$ and  $\alpha=10$ of model \eqref{P2}. As long as $\alpha \ge 10$, model \eqref{P2} is as good as model \eqref{P1} up to a negligible difference.

\item The curve of $\alpha=1$ is noticeably lower than others when the signal is flat or decays slowly. For this reason, we do not recommend using $\alpha = \|\vx^0\|_\infty$ for model \eqref{P2} unless when the underlying signals decay very fast.

\item The differences of  the fours curves are very similar across the two levels $10^{-3}$ and $10^{-5}$ of relative errors. We tested other levels and found the same. Therefore, the performance differences are independent of the error level chosen to plot the curves.
\end{itemize}

Some expert readers may know that in theory, given matrix $\vA$, whether or not model \eqref{P1} can exactly recover $\vx^0$ solely depends on $\sign(\vx^0)$, independent of its decay speed. So, one may wonder why the BP curves are not the same across different plots. That is because, when  \eqref{P1} fails to recover  $\vx^0$, the relative error depends on the decay speed; a faster decaying signal, when not exactly recovered, tends to have a smaller error.
This is why at the error level  $10^{-3}$, the BP curve is  obviously higher (better) on the  faster-decaying signals.

%



\subsection{Null space property}
Matrix  $\vA$ satisfies the NSP if
\begin{align}
\label{nsp}
\|\vh_{\cS}\|_1 <\,\|\vh_{\cS^c}\|_1,
\end{align}
holds for all $\vh\in\Null(\vA)$ and coordinate sets $\cS\subset\{1, 2, \cdots, n\}$ of cardinality $|\cS|\le k$. If so,  problem \eqref{P1}
recovers \emph{all} $k$-sparse vectors $\vx^0$ from measurements $\vb=\vA\vx^0$.
The NSP is also necessary for exact recovery of all  $k$-sparse vectors uniformly. The wide use of NSP can be found in, e.g., \cite{Donoho-Huo-01,Gribonval-Nielsen-03,Zhang-TR05-09}. Note that it holds regardless the value of $\|\vx^0\|_\infty$. We now give a necessary and sufficient condition
for problem \eqref{P2}.

\begin{theorem}[NSP condition]
\label{thm:nsp} 
Assume $\|\vx^0\|_\infty$ is fixed.  Problem  \eqref{P2}  uniquely recovers all $k$-sparse vectors $\vx^0$ with the fixed $\|\vx^0\|_\infty$ from measurements
$\vb=\vA\vx^0$ if and only if
\begin{align}
\label{nsp1}
\left(1 + \frac{\|\vx^0\|_\infty}{\alpha}\right)\|\vh_{\cS}\|_1 \le&\, \|\vh_{\cS^c}\|_1,
\end{align}
holds for all vectors $\vh\in\Null(\vA)$ and coordinate sets $\cS$ of cardinality $|\cS|\le k$.
\end{theorem}
\begin{proof}
\textbf{Sufficiency}: Pick any $k$-sparse vector $\vx^0$. Let $\cS :=\mathrm{supp}(\vx^0)$ and $\cZ = \cS^c$.
For any \emph{nonzero} $\vh\in\mathrm{Null}(\vA)$, we have $\vA(\vx^0+\vh)=\vA\vx^0=\vb$ and
\begin{eqnarray}\nonumber
\|\vx^0+\vh\|_1 + \frac{1}{2\alpha}\|\vx^0+\vh\|_2^2 & = & \|\vx^0_\cS + \vh_\cS\|_1 +
\frac{1}{2\alpha}\|\vx^0_\cS +\vh_\cS\|_2^2 + \|\vh_\cZ\|_1 + \frac{1}{2\alpha}\|\vh_\cZ\|_2^2\\
\nonumber
& \ge & \|\vx^0_\cS\|_1 - \|\vh_\cS\|_1 + \frac{1}{2\alpha}\|\vx^0_\cS\|_2^2 + \frac{1}{\alpha}\langle\vx^0_\cS, \vh_\cS \rangle
+ \frac{1}{2\alpha}\|\vh_\cS\|_2^2 + \|\vh_\cZ\|_1 + \frac{1}{2\alpha}\|\vh_\cZ\|_2^2\\
\nonumber
& \ge & \left[\|\vx^0_\cS\|_1+\frac{1}{2\alpha}\|\vx^0_\cS\|_2^2\right] + \left[\|\vh_\cZ\|_1 - \|\vh_\cS\|_1
- \frac{\|\vx^0_\cS\|_\infty}{\alpha}\|\vh_\cS\|_1\right] + \frac{1}{2\alpha}\|\vh\|_2^2\\
\label{xh12}
& = & \left[\|\vx^0\|_1+\frac{1}{2\alpha}\|\vx^0\|_2^2\right] + \left[\|\vh_\cZ\|_1 - \left(1
+ \frac{\|\vx^0\|_\infty}{\alpha}\right)\|\vh_\cS\|_1\right] + \frac{1}{2\alpha}\|\vh\|_2^2,
\end{eqnarray}
where the first inequality follows from the triangle inequality,  and the second  follows from $\|\vh_\cS\|_2^2+\|\vh_\cZ\|_2^2=\|\vh\|_2^2$ and
$\langle \vx^0_\cS,\vh_\cS \rangle\ge -\|\vx^0_\cS\|_\infty\|\vh_\cS\|_1= -\|\vx^0\|_\infty\|\vh_\cS\|_1$.

Since $\|\vh\|_2^2>0$,  $\|\vx^0+\vh\|_1 + \frac{1}{2\alpha}\|\vx^0+\vh\|_2$ is strictly larger than $\|\vx^0\|_1+\frac{1}{2\alpha}\|\vx^0\|_2$
provided that the second block of \eqref{xh12} is nonnegative. Hence, condition \eqref{nsp1} is sufficient for $\vx^0$ to be
the unique minimizer of \eqref{P2} .

\textbf{Necessity:} It is sufficient to show that  for any given nonzero $\vh\in\Null(\vA)$ and $\cS$ satisfying $|\cS|\le k$, we can to identify a $k$-sparse $\vx^0$ such that
\eqref{nsp1} is \emph{necessary} for its exact recovery. To this end, we define $\vx^0$ as $x^0_i=-\sign(h_i)\|\vh\|_\infty$ for $i\in\cS$ and
$x^0_j=0$ for $j\in \cS^c$, and scale $\vx^0$ to have the specified $\|\vx^0\|_\infty$.  Under this construction, we have the following properties: $\|\vx^0\|_0\le k$, $\|\vx^0_\cS + \tau\vh_\cS\|_1=\|\vx^0_\cS\|_1
- \|\tau\vh_\cS\|_1$, and $\langle \vx^0_\cS,\tau\vh_\cS \rangle= -\|\vx^0_\cS\|_\infty\|\tau\vh_\cS\|_1$,  for any $0<\tau\le 1$. Now, we let  $\tau\vh$ replace $\vh$ in the
 equation array \eqref{xh12} and observe that both of the two inequalities of \eqref{xh12} now hold with equality. Therefore, since the exact recovery of $\vx^0$ requires $\|\vx^0+\tau\vh\|_1 + \frac{1}{2\alpha}\|\vx^0+\tau\vh\|_2^2> \|\vx^0\|_1+\frac{1}{2\alpha}\|\vx^0\|_2^2$, it also requires
\beq
\left[\|\tau\vh_\cZ\|_1 - \left(1
+ \frac{\|\vx^0\|_\infty}{\alpha}\right)\|\tau\vh_\cS\|_1\right] + \frac{1}{2\alpha}\|\tau\vh\|_2^2>0
\eeq
for all $0<\tau\le 1$, which in turn requires \eqref{nsp1} to hold.
\end{proof}
\begin{remark} For any finite $\alpha>0$, \eqref{nsp1} is stronger than \eqref{nsp} due to the extra term  $\frac{\|\vx^0_\cS\|_\infty}{\alpha}$.
Since various uniform  recovery results establish conditions that guarantee
\eqref{nsp}, one can tighten these conditions so that
they  guarantee \eqref{nsp1} and thus the uniform recovery by problem \eqref{P2}. How much  tighter these conditions have to be depends on
the value $\frac{\|\vx^0_\cS\|_\infty}{\alpha}$. 
\end{remark}

\subsection{Restricted isometry property}\label{sc23}
In this subsection, we first review  the RIP-based sparse recovery guarantees and then show that given certain RIP conditions, any $\alpha\ge 10\|\vx^0\|_2$ guarantees exact and stable recovery by \eqref{P2} and \eqref{P3}, respectively.

\begin{definition}\cite{Candes-Tao-05}
The RIP constant $\delta_{k}$ of matrix $\vA$ is the smallest value such that
\begin{equation}
\label{RIP}
(1- \delta_k) \|{\vx}\|_2^2 \le \|A {\vx}\|_2^2 \le (1+ \delta_k) \|{\vx}\|_2^2
\end{equation}
holds for all $k$-sparse vectors ${\vx}\in \mathbb{R}^n$.
\end{definition}

For (\ref{P1}) to recover any $k$-sparse vector uniformly,   \cite{Candes-rip08} shows the sufficiency of $\delta_{2k}< 0.4142$, which is
later improved to $\delta_{2k}< 0.4531$ \cite{Foucart-Lai-09}, $\delta_{2k}<0.4652$ \cite{Foucart-rip-10}, $\delta_{2k}<0.4721$ \cite{CWX10}, as well as
 $\delta_{2k} < 0.4931$ \cite{Mo-Li-11}. The bound is still being improved. 
Adapting results in \cite{Mo-Li-11}, we give the uniform recovery conditions for  \eqref{P2} below.
\begin{theorem}[RIP condition for exact recovery]
\label{main1}
Assume  that $\vx^0\in\mathbb{R}^n$ is $k$-sparse. If $\vA$ satisfies RIP with $\delta_{2k} \le 0.4404$ and $\alpha\ge 10\|\vx^0\|_\infty$,
then $\vx^0$ is the unique minimizer of  \eqref{P2} given measurements $\vb:=\vA\vx^0$.
\end{theorem}
\begin{proof} Let $\cS :=\mathrm{supp}(\vx^0)$ and $\cZ := \cS^c$. Theorem 3.1 in \cite{Mo-Li-11}
shows that any $\vh\in\mathrm{Null}(\vA)$ satisfies
$$
\|\vh_\cS\|_1 \le \theta_{2k} \|\vh_\cZ\|_1,
$$
where
\begin{equation}\label{theta}
\theta_{2k}:=\sqrt{\frac{4(1+5 \delta_{2k} - 4\delta^2_{2k})}{(1-\delta_{2k})(32 - 25 \delta_{2k})}}
\end{equation}
Hence, \eqref{nsp1} holds provided that
$$
\left(1 + \frac{\|\vx^0\|_\infty}{\alpha}\right)^{-1}
\ge\theta_{2k}
$$
or, in light of $\theta_{2k}<1$,
\begin{equation}\label{ripc2}
{\alpha} \ge \left(\theta_{2k}^{-1}-1\right)^{-1}\|\vx^0\|_\infty = \frac{\|\vx^0\|_\infty  \cdot \sqrt{4(1+5 \delta_{2k} - 4\delta^2_{2k})}}{\sqrt{(1-
\delta_{2k})(32 - 25 \delta_{2k})}-\sqrt{4(1+5 \delta_{2k} - 4\delta^2_{2k})}} .
\end{equation}
For $\delta_{2k}=0.4404$, we obtain $ \left(\theta_{2k}^{-1}-1\right)^{-1}\|\vx^0\|_\infty\approx 9.9849 \|\vx^0\|_\infty\le \alpha$, which
proves the theorem.
\end{proof}

\begin{remark} Different values of $\delta_{2k}$ are associated with different  conditions on $\alpha$. Following \eqref{ripc2}, if $\delta_{2k}\le 0.4715$,
$\alpha \ge 25\|\vx^0\|_\infty$ guarantees exact recovery. If  $\delta_{2k}\le 0.1273$, $\alpha\ge \|\vx^0\|_\infty$ guarantees exact recovery. In
general, a smaller $\delta_{2k}$ allows a smaller $\alpha$. \end{remark}

Next we study the case where $\vb$ is noisy or $\vx^0$ is not exactly  sparse, or both. For comparison, we present two inequalities next to each other for problems \eqref{P1n} and \eqref{P2} each, where the first one is easy to obtain; see \cite{Candes-rip08} for example.
\begin{lemma}\label{lmop}
Let $\vx^0\in\mathbb{R}^n$ be an \emph{arbitrary}  vector,  $\cS$ be the coordinate set of its $k$ largest components in magnitude,
and $\cZ := \{1, \cdots, n\}\setminus \cS$. Let $\bar{\vx}^*$ and $\vx^*$ be the solutions of \eqref{P1n} and \eqref{P3}, respectively.
The error vectors $\bar{\vh} = \bar{\vx}^*-\vx^0$ and $\vh=\vx^*-\vx^0$  satisfy
\begin{align}
\label{err1}
\|\bar{\vh}_\cZ\|_1 \le &\hspace{15pt} \|\bar{\vh}_\cS\|_1 +\hspace{6.8pt}2\|\vx^0_\cZ\|_1, \\
\label{err2}
\|\vh_\cZ\|_1 \le & ~C_3 \|\vh_\cS\|_1 +C_4\|\vx^0_\cZ\|_1,
\end{align}
where $\|\vx^0_\cZ\|_1$ is the best $k$-term approximation error of $\vx^0$ and
\begin{equation}\label{C34}
C_3  :=  \frac{\alpha + \|\vx^0_\cS\|_\infty}{\alpha -\|\vx^0_\cZ\|_\infty}\quad\hbox{and}\quad C_4  :=  \frac{2\alpha}{\alpha -\|\vx^0_\cZ\|_\infty}.
\end{equation}
\end{lemma}
\begin{proof} 
We only show \eqref{err2}. Since $\vx^* = \vx^0+\vh$ is the minimizer of \eqref{P3}, we have
\begin{equation}
\label{xhh1}
\|\vx^0+\vh\|_1 + \frac{1}{2\alpha}\|\vx^0+\vh\|_2^2 \le \|\vx^0\|_1 + \frac{1}{2\alpha}\|\vx^0\|_2^2.
\end{equation}
Also,
\begin{eqnarray}
\nonumber
\|\vx^0+\vh\|_1 + \frac{1}{2\alpha}\|\vx^0+\vh\|_2^2 & = & \|\vx^0_\cS+\vh_\cS\|_1 + \frac{1}{2\alpha}\|\vx^0_\cS+\vh_\cS\|_2^2 + \|\vx^0_\cZ+\vh_\cZ\|_1
+ \frac{1}{2\alpha}\|\vx^0_\cZ+\vh_\cZ\|_2^2\\
\nonumber
& \ge & \|\vx^0_\cS\|_1 - \|\vh_\cS\|_1 + \frac{1}{2\alpha}\|\vx^0_\cS\|_2^2 - \frac{1}{\alpha}|\langle \vx^0_\cS,\vh_\cS \rangle|
+ \frac{1}{2\alpha}\|\vh_\cS\|_2^2\\
\nonumber
& & + \|\vh_\cZ\|_1 - \|\vx^0_\cZ\|_1 + \frac{1}{2\alpha}\|\vx^0_\cZ\|_2^2 - \frac{1}{\alpha}|\langle \vx^0_\cZ,\vh_\cZ \rangle| +
\frac{1}{2\alpha}\|\vh_\cZ\|_2^2\\
\nonumber
& = & (\|\vx^0\|_1 + \frac{1}{2\alpha}\|\vx^0\|_2^2) - 2\|\vx^0_\cZ\|_1 - (\|\vh_\cS\|_1 + \frac{1}{\alpha}|\langle \vx^0_\cS,\vh_\cS \rangle|)\\
\nonumber
& &  + ( \|\vh_\cZ\|_1  - \frac{1}{\alpha}|\langle \vx^0_\cZ, \vh_\cZ \rangle| ) + \frac{1}{2\alpha}\|\vh\|_2^2\\
\nonumber
&\ge & (\|\vx^0\|_1 + \frac{1}{2\alpha}\|\vx^0\|_2^2) - 2\|\vx^0_\cZ\|_1-\left(1+\frac{\|\vx_\cS^0\|_\infty}{\alpha}\right)\|\vh_\cS\|_1\\
\label{xhh2}
& & + \left( 1 - \frac{\|\vx^0_\cZ\|_\infty}{\alpha}\right) \|\vh_\cZ\|_1 + \frac{1}{2\alpha}\|\vh\|_2^2,
\end{eqnarray}
where the first inequality follows from the triangle inequality, and the second from $\langle\va,\vb\rangle \le \|\va\|_\infty\|\vb\|_1$.
Combining \eqref{xhh1} and \eqref{xhh2}, we obtain
$$   \left( 1 - \frac{\|\vx^0_\cZ\|_\infty}{\alpha}\right) \|\vh_\cZ\|_1+ \frac{1}{2\alpha}\|\vh\|_2^2 \le \left(1+\frac{\|\vx_\cS^0\|_\infty}{\alpha}\right)\|\vh_\cS\|_1 + 2\|\vx^0_\cZ\|_1 $$
and thus \eqref{err2} after dropping the nonnegative  term $ \frac{1}{2\alpha}\|\vh\|^2$.
\end{proof}
We now present the stable recovery guarantee.
\begin{theorem}[RIP condition for stable recovery]
\label{main2}
Assume the setting of Lemma \ref{lmop}.
Let $\vb:=A \vx^0 + \vn$, where $\vn$ is an \emph{arbitrary} noisy vector with $\|\vn\|_2\le \sigma$.
If $\vA$ satisfies RIP with $\delta_{2k} \le 0.3814$,
then the solution $\vx^*$ of \eqref{P3} with any $\alpha\ge 10\|\vx^0\|_\infty$ satisfies
\begin{align}
\label{bnd1}
\|\vx^* - \vx^0\|_1  \le & C_1 \cdot \sqrt{k}\|\vn\|_2 + C_2 \cdot\|\vx^0_\cZ\|_1,\\
\label{bnd2}
\|\vx^* - \vx^0\|_2  \le & \bar{C}_1 \cdot \|\vn\|_2 + \bar{C}_2 \cdot\|\vx^0_\cZ\|_1/\sqrt{k},
\end{align}
where $C_1$, $C_2$, $\bar{C}_1$, and $\bar{C}_2$ are given in \eqref{sc1}--\eqref{sbc2} as functions of only $\delta_{2k}$, $C_3$,
and $C_4$ in \eqref{C34}.
\end{theorem}

\begin{proof}
We follow an argument similar to that in \cite{Mo-Li-11}.
According to Lemma 4.3 of
\cite{Mo-Li-11}, from $\|\vA\vh\|_2=\|\vA\vx^* - \vA\vx^0\|_2=\|\vA\vx^*-\vb+\vn\|_2\le \|A\vx^*-\vb\|_2+\|\vn\|_2\le 2\|\vn\|_2$ and
$\delta_{2k}<2/3$, we obtain
\begin{equation}
\label{hsz2}
\|\vh_\cS\|_1 \le \frac{2\sqrt{2}}{\sqrt{1-\delta_{2k}}}\sqrt{k}\|\vn\|_2+ \theta_{2k}\|\vh_\cZ\|_1,
\end{equation}
where $\theta_{2k}$ is defined in \eqref{theta} as a function of $\delta_{2k}$. It is easy to verify that with the choice of $\delta_{2k}\le 0.3814$ and $\alpha$
 in the theorem, $C_3\theta_{2k}<1$ holds for all nonzero $\vx^0$. 
Hence, combining \eqref{err2} of Lemma \ref{lmop}
and \eqref{hsz2} yield the bound of $\|\vh_\cZ\|_1$:
\begin{equation}
\label{hsz3}
\|\vh_\cZ\|_1 \le (1-C_3\theta_{2k})^{-1}\left(C_3 \frac{2\sqrt{2}}{\sqrt{1-\delta_{2k}}}\sqrt{k}\|\vn\|_2 +  C_4\|\vx^0_\cZ\|_1\right).
\end{equation}
Applying \eqref{hsz2} and \eqref{hsz3} gives us \eqref{bnd1} or
\begin{eqnarray*}
\|\vx^* - \vx^0\|_1=\|\vh\|_1 &=& \|\vh_\cS\|_1+ \|\vh_\cZ\|_1 \\
&\le & \frac{2\sqrt{2}}{\sqrt{1-\delta_{2k}}}\sqrt{k}\|\vn\|_2+ (1+ \theta_{2k}) \|\vh_\cZ\|_1\\
&\le &C_1\sqrt{k}\|\vn\|_2 + C_2\|\vx^0 - \sigma_k(x^0)\|_1,
\end{eqnarray*}
where
\begin{subequations}
\begin{eqnarray}
\label{sc1}
C_1 & = & \frac{2\sqrt{2}(1+C_3)}{\sqrt{1-\delta_{2k}}(1-C_3\theta_{2k})},\\
\label{sc2}
C_2 & = & \frac{(1+\theta_{2k}) C_4}{1-C_3\theta_{2k}}.
\end{eqnarray}
\end{subequations}
To prove \eqref{bnd2},  we apply  \eqref{hsz3} to the inequality (Page 7 of \cite{Mo-Li-11})
$$
\|\vh\|_2 \le \frac{2}{\sqrt{1-\delta_{2k}}}\|\vn\|_2 + \sqrt{\frac{8(2-\delta_{2k})}{(1-\delta_{2k})(32-25\delta_{2k})}}\cdot\frac{\|\vh_\cZ\|_1}{\sqrt{k}},
$$
and obtain \eqref{bnd2} or
$$
\|x^* - x^0\|_2=\|\vh\|_2 \le \bar{C}_1\|\vn\|_2 + \bar{C}_2\|\vx^0 - x^0_{[k]}\|_1/\sqrt{k},
$$
where
\begin{subequations}
\begin{eqnarray}
\label{sbc1}
\bar{C}_1 & := & \frac{2}{\sqrt{1-\delta_{2k}}}\left(\frac{4C_3}{1-C_3\theta_{2k}}\sqrt{\frac{2-\delta_{2k}}{(1-\delta_{2k})(32-25\delta_{2k})}}+1\right),\\
\label{sbc2}
\bar{C}_2 & := & \frac{2C_4}{1-C_3\theta_{2k}}\sqrt{\frac{2(2-\delta_{2k})}{(1-\delta_{2k})(32-25\delta_{2k})}}.
\end{eqnarray}
\end{subequations}
\end{proof}
\begin{remark}
A key inequality in the proof above is  $C_3\theta_{2k}<1$, where $C_3$ (cf. \eqref{C34}) depends on $\alpha$, $\|\vx^0_\cS\|_\infty$, and
$\|\vx^0_\cZ\|_\infty$, and $\theta_{2k}$ (cf. \eqref{theta}) depends on $\delta_{2k}$. If the nonzeros of $\vx^0$ decay faster in magnitude, $C_3$
becomes smaller
and thus the condition $C_3\theta_{2k}<1$ is \emph{easier} to hold. Therefore, a faster decaying $\vx^0$ is easier to recover. This is consistent with the numerical simulation in subsection \ref{motv}. In Theorem \ref{main2}, the
condition on $\delta_{2k}$ and bound on $\alpha$ are given
for the worst case corresponding to no decay, namely, $\|\vx^0_\cS\|_\infty=\|\vx^0_\cZ\|_\infty$. If $\|\vx^0_\cS\|_\infty >\|\vx^0_\cZ\|_\infty$, one can allow a
larger $\delta_{2k}$ for each fixed $\alpha$ or, equivalently, a smaller $\alpha$ for each fixed $\delta_{2k}$. For example, if $\|\vx^0_\cS\|_\infty
\ge 10 \|\vx^0_\cZ\|_\infty$, one only needs $\delta_{2k}\le 0.4348$ instead of the theorem-assumed condition  $\delta_{2k}\le 0.3814$.

There is also a trade-off between $\delta_{2k}$ and $\alpha$.  Under the worst case $\|\vx^0_\cS\|_\infty=\|\vx^0_\cZ\|_\infty$, imposing to
$\alpha\ge 25\|\vx^0\|_\infty$ leads to the relaxed condition    $\delta_{2k}\le 0.4489$. 
\end{remark}

\subsection{Spherical section property}

Next, we derive exact and stable recovery conditions based on the spherical section property (SSP) \cite{Zhang-08,Vavasis-09} of $\vA$, which has the advantage
of invariance to left-multiplying nonsingular matrices to the sensing matrix $\vA$, as pointed out in \cite
{Zhang-08}. On the other hand, more matrices are known to satisfy the  RIP than the SSP.
\begin{definition}[$\Delta$-SSP \cite{Vavasis-09}] Let $m$ and $n$ be two  integers such that $m>0,~n>0$, and $m<n$. An $(n-m)$ dimensional subspace $\cV\subset
\mathbb{R}^n$ has the $\Delta$ spherical section property if
\beq\label{ssd}
\frac{\|\vh\|_1}{\|\vh\|_2} \ge \sqrt{\frac{m}{\Delta}}
\eeq
holds for all nonzero $\vh\in \cV$.
\end{definition}
To see the significance of \eqref{ssd}, we note that (i) $\frac{\|\vh\|_1}{\|\vh\|_2} \ge 2\sqrt{k}$ for all $\vh\in\Null(\vA)$ is a sufficient
condition for the NSP inequality \eqref{nsp} and (ii) due to \cite{Kashin-77,Garnaev-Gluskin-84}, a uniformly random $(n-m)$-dimensional subspace
$\cV\subset\RR^n$ has the SSP for
$$\Delta = C_0(\log(n/m)+1)$$
with probability at least $1-\exp(C_1(n-m))$, where $C_0$ and $C_1$ are
universal constants. Hence, $m> 4k\Delta$ guarantees \eqref{nsp} to hold, and furthermore, if $\Null(\vA)$ is uniformly random, $m=O(k\log(n/m))$ is
sufficient for \eqref{nsp} to hold with overwhelming probability \cite{Zhang-08,Vavasis-09}. 
These results can be extended to the augmented model \eqref{P2}.

\begin{theorem}[SSP condition for exact recovery]\label{thm:ssp}
Suppose $\Null(\vA)$ satisfies the  $\Delta$-SSP. Let us fix $\|\vx^0\|_\infty$ and  $\alpha>0$. If
\beq\label{ssd1}
m\ge   \left(2+ \frac{\|\vx^0\|_\infty}{\alpha}\right)^2k \Delta,
\eeq
then the null-space condition \eqref{nsp1} holds for all $\vh\in\Null(\vA)$ and coordinate sets $\cS$ of cardinality $|\cS|\le k$. By Theorem
\ref{thm:nsp}, \eqref{ssd1} guarantees that problem \eqref{P2} recovers any $k$-sparse $\vx^0$ from measurements $\vb=\vA\vx^0$.
\end{theorem}
\begin{proof}
Let $\cS$ be a coordinate set with $|\cS|\le k$. Condition \eqref{nsp1} is equivalent to
\beq\label{nsp2}
\left(2 + \frac{\|\vx^0_\cS\|_\infty}{\alpha}\right)\|\vh_{\cS}\|_1 \le \|\vh\|_1,
\eeq
Since $\|\vh_\cS\|_1\le \sqrt{k}\|\vh_\cS\|_2\le\sqrt{k}\|\vh\|_2$, \eqref{nsp2} holds provided that
\beq
\left(2
+ \frac{\|\vx^0\|_\infty}{\alpha}\right)\sqrt{k}\le \frac{\|\vh\|_1}{\|\vh\|_2},
\eeq
which  itself holds, in light of \eqref{ssd}, provided that \eqref{ssd1} holds.
\end{proof}
Now we consider the case $\vA \vx^0=\vb$ where $\vx^0$ is an approximately sparse vector.
\begin{theorem}[SSP condition for stable recovery]\label{thm:ssp1} Suppose $\Null(\vA)$ satisfies the  $\Delta$-SSP. Let $\vx^0\in\mathbb{R}^n$ be an \emph{arbitrary}  vector,  $\cS$ be the coordinate set of
its $k$ largest components in magnitude, and $\cZ := \{1, \cdots, n\}\setminus \cS$. Let $\alpha>0$ in problem \eqref{P2}. Let $C_3$ and $C_4$ be defined in
\eqref{C34}, which depend on $\alpha$. If
\beq\label{ssd2}
m\ge   4\left(1+ C_3\right)^2k \Delta,
\eeq
 then the solution $\vx^*$ of \eqref{P2} satisfies
\beq\label{vx4}
\|\vx^*-\vx^0\|_1\le 4C_4\|\vx_{\cZ}^0 \|_1,
\eeq
where $\|\vx_{\cZ}^0 \|_1$ is the best $k$-term approximation error of $\vx^0$.
\end{theorem}
\begin{proof}
Let $\vh=\vx^*-\vx^0\in\Null(\vA)$. Let
\beq\label{bc}
\bar{C}=\frac{\|\vh\|_1}{\|\vx_\cZ^0\|_1}.
\eeq
Then \eqref{vx4} is equivalent to
\beq\label{acc}
\bar{C}\le 4C_4.
\eeq
Adding $\|\vh_\cS\|_1$ to \eqref{err2} and plugging in \eqref{bc} gives us
\beq
\|\vh\|_1 \le (1+C_3) \|\vh_\cS\|_1 +2C_4\bar{C}^{-1}\|\vh\|_1,
\eeq
or $(1-2C_4\bar{C}^{-1})\|\vh\|_1\le (1+C_3)\|\vh_\cS\|_1$. If $\bar{C}\le 2C_4$, \eqref{acc} naturally holds. Otherwise, we
have $ \bar{C}> 2C_4$ and \beq\label{d2c4}
\|\vh\|_1\le\frac{1+C_3}{1-2C_4\bar{C}^{-1}}\|\vh_\cS\|_1\le \frac{(1+C_3)\sqrt{k}}{1-2C_4\bar{C}^{-1}}\|\vh\|_2.
\eeq
Now, combining $\Delta$-SSP and \eqref{ssd2}, we obtain
\beq
\frac{\|\vh\|_1}{\|\vh\|_2}\ge \sqrt{\frac{m}{\Delta}}\ge 2\left(1+ C_3\right)\sqrt{k},
\eeq
which together with \eqref{d2c4} gives \eqref{acc}.
\end{proof}


\subsection{``{RIPless}'' analysis}
The ``{RIPless}'' analysis \cite{Candes-Plan-10} gives non-uniform recovery guarantees for a wide class of compressive sensing matrices such as those with iid subgaussian entries, orthogonal transform ensembles satisfying an incoherence condition, random Toeplitz/circulant ensembles, as well as certain tight and continuous frame ensembles, at $O(k\log(n))$ measurements. This analysis is especially useful in situations where the RIP, as well as NSP and SSP, is difficult to check or does not hold. In this subsection, we describe how to adapt the ``RIPless'' analysis to model \eqref{P2}.

\begin{theorem}[RIPless for exact recovery]\label{thm:ripless}
Let $\vx^0\in\RR^n$ be a fixed $k$-sparse vector. With probability at least $1-5/n-e^{-\beta}$, $\vx^0$ is the unique solution to problem \eqref{P2} with $\vb=\vA \vx^0$ and $\alpha \ge 8\|\vx^0\|_2$ as long as the number of measurements
$$
m\ge C_0(1+\beta)\mu(\vA)\cdot k\log n,
$$
where $C_0$ is a universal constant and $\mu(\vA)$ is the incoherence parameter of $\vA$ (see \cite{Candes-Plan-10} for its definition and values for various kinds of compressive sensing matrices).
\end{theorem}
\begin{proof}
The proof is mostly the same as that of Theorem 1.1 of \cite{Candes-Plan-10} except we shall adapt Lemma 3.2 of \cite{Candes-Plan-10}  to Lemma \ref{lm:dc} below for our model \eqref{P2}. We describe the proof of the theorem very briefly here. For any matrix $\vA$ satisfying property \eqref{vc1} in Lemma \ref{lm:dc}, the golfing scheme \cite{Gross-09} can be used to construct a dual vector $\vy$ such that $\vA^*\vy$ satisfies property \eqref{vc2} in Lemma \ref{lm:dc}.  The properties \eqref{vc1} and \eqref{vc2} and the construction are exactly the same as in \cite{Candes-Plan-10}. Then Lemma \ref{lm:dc} below lets this $\vA^*\vy$  guarantee the optimality of $\vx^0$ to \eqref{P3}. 
\end{proof}
\begin{lemma}[Dual certificate]\label{lm:dc}
Let $\vx^0$ be given in Theorem \ref{thm:ripless} and $\cS:=\supp(\vx^0)$. If $\vA=[\va_1~\va_2~\cdots~\va_n]$ satisfies
\beq\label{vc1}
\|(\vA^*_\cS \vA_\cS)^{-1}\|_2\le 2\quad\text{and}\quad\max_{i\in\cS^c}\|\vA^*_\cS \va_i\|_2\le 1
\eeq
and there exists $\vy$ such that  $\vv=\vA^*\vy$ satisfies
\beq\label{vc2}
\|\vv_\cS - \sign(\vx^0_\cS)\|_2\le 1/4\quad\text{and}\quad\|\vv_{\cS^c}\|_\infty\le 1/4,
\eeq
then $\vx^0$ is the unique solution to \eqref{P2} with $\vb=\vA \vx^0$ and $\alpha \ge 8\|\vx^0\|_2$.
\end{lemma}
\begin{proof}
Let $\cZ:=\cS^c$. For any \emph{nonzero} $\vh\in\mathrm{Null}(\vA)$, we have $\vA\vh=\vzero$ and
\begin{eqnarray}\nonumber
\|\vx^0+\vh\|_1 + \frac{1}{2\alpha}\|\vx^0+\vh\|_2^2 & = & \|\vx^0_\cS + \vh_\cS\|_1 +
\frac{1}{2\alpha}\|\vx^0_\cS +\vh_\cS\|_2^2 + \|\vh_\cZ\|_1 + \frac{1}{2\alpha}\|\vh_\cZ\|_2^2\\
\nonumber
& \ge & \|\vx^0_\cS\|_1 +\langle\sign(\vx_\cS),\vh_\cS\rangle + \frac{1}{2\alpha}\|\vx^0_\cS\|_2^2 + \frac{1}{\alpha}\langle\vx^0_\cS, \vh_\cS \rangle
+ \frac{1}{2\alpha}\|\vh_\cS\|_2^2 + \|\vh_\cZ\|_1 + \frac{1}{2\alpha}\|\vh_\cZ\|_2^2\\
\nonumber
& \ge & \left[\|\vx^0_\cS\|_1+\frac{1}{2\alpha}\|\vx^0_\cS\|_2^2\right] + \left[ \langle\sign(\vx_\cS),\vh_\cS\rangle
+\frac{1}{\alpha}\langle\vx^0_\cS, \vh_\cS \rangle+\|\vh_\cZ\|_1 \right] + \frac{1}{2\alpha}\|\vh\|_2^2\\
\label{xh13}
& = & \left[\|\vx^0\|_1+\frac{1}{2\alpha}\|\vx^0\|_2^2\right] + \left[ \langle\sign(\vx_\cS),\vh_\cS\rangle
+\frac{1}{\alpha}\langle\vx^0_\cS, \vh_\cS \rangle+\|\vh_\cZ\|_1 \right] + \frac{1}{2\alpha}\|\vh\|_2^2
\end{eqnarray}
Since the last term of \eqref{xh13} is strictly positive, $\vx^0$ is the unique solution to \eqref{P2} provided that
\beq
\label{ned}
\langle\sign(\vx_\cS),\vh_\cS\rangle
+\frac{1}{\alpha}\langle\vx^0_\cS, \vh_\cS \rangle+\|\vh_\cZ\|_1\ge 0.
\eeq
Following the proof of Lemma 3.2 in \cite{Candes-Plan-10} and from \eqref{vc1} and \eqref{vc2} we obtain
$$
\langle\sign(\vx_\cS),\vh_\cS\rangle\ge-\frac{1}{4}\left(\|\vh_\cS\|_2+\|\vh_\cZ\|_1\right)\quad\text{and}\quad \|\vh_\cZ\|_1\ge\frac{1}{2}\|\vh_\cS\|_2,
$$
which together with $\alpha\ge 8\|\vx^0\|_2$ give
\begin{align*}
\langle\sign(\vx_\cS),\vh_\cS\rangle
+\frac{1}{\alpha}\langle\vx^0_\cS, \vh_\cS \rangle+\|\vh_\cZ\|_1 & \ge -\frac{1}{4}\left(\|\vh_\cS\|_2+\|\vh_\cZ\|_1\right )-\frac{\|\vx^0_\cS\|_2}{\alpha}\|\vh_\cS\|_2+\|\vh_\cZ\|_1\\
& \ge -\frac{1}{4}\|\vh_\cS\|_2 + \frac{3}{4}\|\vh_\cZ\|_1 - \frac{1}{8}\|\vh_\cS\|_2\\
& \ge \frac{3}{8}\|\vh_\cS\|_2- \frac{3}{8}\|\vh_\cS\|_2\\
& = 0.
\end{align*}
Hence, $\vx^0+\vh$ gives a strictly worse objective \eqref{P2} than $\vx^0$, so $\vx^0$ is the unique solution to \eqref{P2}.
\end{proof}

\subsection{Matrix Recovery Guarantees}
It is fairly easy to extend the results above, except the ``RIPless'' analysis, to the recovery of low-rank matrices. Throughout this subsection, we let $\sigma_i(\vX),~ i=1, \cdots, m$ denote the $i$th largest  singular value of matrix $\vX$ of rank $m$ or less, and let $\|\vX\|_*:=\sum_{i=1}^m\sigma_i(\vX)$, $\|\vX\|_F:=\left(\sum_{i=1}^m\sigma^2_i(\vX)\right)^{1/2}$, and $\|\vX\|_2=\sigma_1(\vX)$ denote the nuclear, Frobenius, and spectral norms of $\vX$, respectively.

The extension is based on the following property of unitarily invariant matrix norms.
\begin{lemma}[\cite{Horn-Johnson-book-90} Theorem 7.4.51]
\label{normproperty}
Let $\vX$ and $\vY$ be two matrices of the same size. Any unitarily invariant  norm $\|\cdot\|_\phi$ satisfies
\begin{equation}
\label{strongform}
\|\Sigma(\vX)-\Sigma(\vY)\|_\phi\leq \|\vX - \vY\|_\phi,
\end{equation}
where $\Sigma(\vX)=\diag(\sigma_1(\vX), \cdots, \sigma_m(\vX))$ and $\Sigma(\vY)=\diag(\sigma_1(\vY), \cdots, \sigma_m(\vY))$ are two diagonal matrices.
\end{lemma}
In particular,  matrices $\vX$ and $\vY$ obey
\begin{equation}
\label{Ineq1}
\sum_{i=1}^m |\sigma_i(\vX)- \sigma_i(\vY)|\le \|\vX- \vY\|_*
\end{equation}
and
\begin{equation}
\label{Ineq2}
\sum_{i=1}^m \left(\sigma_i(\vX) - \sigma_i(\vY)\right)^2 \le \|\vX - \vY\|^2_F.
\end{equation}


%
%
%

By  applying \eqref{Ineq1}, \cite{OMFH11} shows that any sufficient conditions based on RIP and SSP  of $\vA$ for recovering sparse vectors by model \eqref{P1} can be translated to  sufficient conditions based on similar properties of $\cA$ for recovering low-rank matrices by model \eqref{N1}. We can establish similar translations from model \eqref{P3} to model \eqref{N3} using both inequalities \eqref{Ineq1} and \eqref{Ineq2}. 
Hence, we present the low-rank matrix recovery results  only with the parts that are \emph{different} from their vector counterparts.

Paper \cite{OH10} presents the NSP condition for problem \eqref{N1}:
all matrices $\vX^0$ of rank $r$ or less can be exactly recovered by problem \eqref{N1} from measurements $\vb=\cA(\vX^0)$ if and only if  all $\vH\in \Null(\mathcal{A})\backslash\{\vzero\}$ satisfy 
\beq\label{nspMC}
\sum_{i=1}^{r}\sigma_i(\vH)<\sum_{i=r+1}^{m}\sigma_i(\vH).
\eeq
We can extend this result to problem (\ref{N2}) by applying inequalities (\ref{Ineq1}) and (\ref{Ineq2}).
\begin{theorem}[Matrix NSP condition]
\label{thm:nsp4MC}
Assume that $\|\vX^0\|_2$ is fixed.  Problem  \eqref{N2}  uniquely recovers all matrices $\vX^0$ (with the specified $\|\vX^0\|_2$) of rank $r$ or less from measurements
$\vb=\cA(\vX^0)$ if and only if
\begin{align}
\label{nspMC1}
\left(1 + \frac{\|\vX^0\|_2}{\alpha}\right)\sum_{i=1}^r\sigma_i(\vH) \le &\, \sum_{i=r+1}^m \sigma_i(\vH)
\end{align}
holds for all matrices $\vH\in\Null(\cA)$.
\end{theorem}
\begin{proof}
\textbf{Sufficiency}: Pick any matrix $\vX^0$ of rank  $r$ or less and let $\vb=\cA (\vX^0)$.
For any \emph{nonzero} $\vH\in\mathrm{Null}(\cA)$, we have $\cA(\vX^0+\vH)=\cA\vX^0=\vb$. By using (\ref{Ineq1}) and (\ref{Ineq2}), we have
\begin{align}
\nonumber
&\|\vX^0+\vH\|_* + \frac{1}{2\alpha}\|\vX^0+\vH\|_F^2  \ge \|s(\vX^0)-s(\vH)\|_1 + \frac{1}{2\alpha} \|s(\vX^0)-s(\vH)\|_2^2 \\
 \ge & \left[\|\vX^0\|_*+\frac{1}{2\alpha}\|\vX^0\|_F^2\right]
+ \left[\sum_{i=r+1}^m \sigma_i(\vH) -  \left( 1+  \frac{\|\vX^0\|_2}{\alpha}\right)\sum_{i=1}^r\sigma_i(\vH)\right] + \frac{1}{2\alpha}\|\vH\|_F^2
\label{xH12}
\end{align}
where the second inequality follows from \eqref{xh12} by letting $\vh = -s(\vH)$ and $\cS=\{1,\ldots,r\}$ and noticing $\vh_\cS=\sum_{i=1}^r\sigma_i(\vH)$ and $\vh_\cZ = \sum_{i=r+1}^m \sigma_i(\vH)$.


For any nonzero $\vH\in\Null(\cA)$, $\|\vH\|_F>0$. Hence, from \eqref{xH12} and \eqref{nspMC1}, it follows that $\vX^0+\vH$ leads to a strictly worse objective than $\vX^0$. That is, $\vX^0$ is the unique solution to problem \eqref{N2}.


\textbf{Necessity:} For any  nonzero $\vH\in\Null(\cA)$ obeying (\ref{nspMC1}), let $\vH=\vU \Sigma \vV^\top$ be the SVD of $\vH$. Construct $\vX^0 = - \vU \Sigma_r \vV^\top$, where $\Sigma_r$ keeps only the largest $r$  diagonal entries of $\Sigma$ and sets the rest to 0.
Scale $\vX^0$ so that it has the specified $\|\vX^0\|_2$. We have
$$
\|\vX^0+t\vH\|_* + \frac{1}{2\alpha}\|\vX^0+t\vH\|_F^2 =
\|\vX^0\|_* + \frac{1}{2\alpha}\|\vX^0\|_F^2 + \left[\sum_{i=r+1}^m\sigma_i(t\vH)- \left(1+ \frac{\|\vX^0\|_2}{\alpha}\right) \sum_{i=1}^r \sigma_i(t\vH)\right]
 + {\frac{1}{2\alpha}\|t\vH\|_F^2}
$$
for any $t>0$. For $\vX^0$ to be the unique solution to \eqref{N2} given $\vb=\cA(\vX^0)$, we must have
$$
\left[\sum_{i=r+1}^m\sigma_i(t\vH)- \left(1+ \frac{\|\vX^0\|_2}{\alpha}\right) \sum_{i=1}^r \sigma_i(t\vH)\right]
 + {\frac{1}{2\alpha}\|t\vH\|^2_F}>0
$$
for all $t>0$. Hence, \eqref{nspMC1} is necessary.
\end{proof}

Paper \cite{Recht-Fazel-Parrilo-07} introduces the following
RIP for matrix recovery.
\begin{definition}[Matrix RIP]
Let $\cM_r:=\{\vX\in\mathbb{R}^{n_1\times n_2}:\rank(\vX)\le r\}$. The RIP constant $\delta_{r}$ of linear operator $\cA$ is the smallest value such
that
\begin{equation}
\label{D-RIP}
(1- {\delta}_r) \|{\vX}\|_F^2 \le \|\cA(\vX)\|_2^2 \le (1+ {\delta}_r) \|{\vX}\|_F^2
\end{equation}
holds for all ${\vX}\in \cM_r$.
\end{definition}

To uniformly recover all matrices of rank $r$ or less by solving \eqref{N1}, it is sufficient for $\cA$ to satisfy $\delta_{5r} <0.1$ \cite{Recht-Fazel-Parrilo-07}, which has been improved to the RIP with $\delta_{4r}<\sqrt{2}-1$ in \cite{Candes-Plan-11} and to $\delta_{2r}<0.307$, as well as ones involving $\delta_{3r}$, $\delta_{4r}$, and $\delta_{5r}$, in \cite{Mohan-Fazel-10}. The algorithm SVP \cite{Meka-Jain-Dhillon-09} provably achieves exact recovery if $\delta_{2r}<1/3$. 

Next, we present a stronger RIP-based condition for the \emph{unsmoothed} problem \eqref{N1}, and then extend it to the \emph{smoothed} problem \eqref{N2} without a proof. 

\begin{theorem}[RIP condition for exact recovery by \eqref{N1}]
\label{mainM0}
Let $\vX^{0}$ be a matrix with rank $r$ or less. Problem \eqref{N1} exactly recovers $\vX^{0}$ from measurements $\vb=\cA(\vX^{0})$ if  $\cA$ satisfies the RIP with $\delta_{2r}<0.4931$.
\end{theorem}

The proof is a straightforward extension to the arguments in \cite{Mo-Li-11} using arguments in \cite{OMFH11}; the interested reader can find it in Appendix. Next we present the result for the augmented model (\ref{N2}).
\begin{theorem}[RIP condition for exact recovery]
\label{mainM1}
Let $\vX^{0}$ be a matrix with rank $r$ or less. The augmented model \eqref{N2} exactly recovers $\vX^{0}$ from measurements $\vb=\cA(\vX^{0})$ if  $\cA$ satisfies the RIP with $\delta_{2r}<0.4404$ and in \eqref{N2} $\alpha\ge 10\|\vX^0\|_2$.
\end{theorem}

\begin{proof} The proof of Theorem \ref{mainM0} in Appendix establishes
that  any $\vH\in\Null(\cA)$ satisfies
$
\|\vH_0\|_* \le \theta_{2r} \|\sum_{i\ge 1}\vH_i\|_*.
$
Hence, \eqref{nspMC1} holds if
$
\left(1 + \frac{\|\vX^0\|_2}{\alpha}\right)^{-1} \ge\theta_{2r}.
$
The rest of the proof is similar to that of Theorem \ref{main1}.
\end{proof}
Skipping a proof similar to that of Theorem \ref{main2}, we present the stable recovery result as follows.
\begin{theorem}[RIP condition for stable recovery]
\label{mainM2}
Let $\vX^0\in\mathbb{R}^{n_1\times n_2}$ be an arbitrary matrix and $\sigma_{i}(\vX^0) $ be its $i$-th largest singular value. Let $\vb: =
\cA(\vX^0)+\vn$, where $\cA$ is a linear operator and $\vn$ is an arbitrary noise vector. If ${\cal A}$
satisfies  the RIP with $\delta_{2r} \le 0.3814$, then the solution $\vX^*$ of \eqref{N3} with any $\alpha\ge 10 \cdot \|\vX^0\|_2$ satisfies the error bounds:
\begin{align}
\label{MCbnd1}
\|\vX^* - \vX^0\|_*  \le & C_1 \cdot \sqrt{k}\|\vn\|_2 + C_2 \cdot\hat{\sigma}(\vX^0),\\
\label{MCbnd2}
\|\vX^* - \vX^0\|_F  \le & \bar{C}_1 \cdot \|\vn\|_2 + (\bar{C}_2 /\sqrt{r}) \cdot\hat{\sigma}(\vX^0),
\end{align}
where $\hat{\sigma}(\vX^0):=\sum_{i=r+1}^{\min\{n_1,n_2\}}\sigma_i(\vX^0)$ is the best rank-$r$ approximation error of $\vX^0$, $C_1$, $C_2$, $\bar{C}_1$, and $\bar{C}_2$ are given by formulas \eqref{sc1}--\eqref{sbc2} in which $\theta_{2k}$ shall be replaced by $\theta_{2r}$ (given in \eqref{t2r}), and
\beq\label{C34m}
C_3 := \frac{\alpha + \sigma_1(\vX^0)}{\alpha - \sigma_{r+1}(\vX^0)}\quad\text{and}\quad C_4 := \frac{2\alpha}{\alpha - \sigma_{r+1}(\vX^0)},
\eeq
respectively.
\end{theorem}
Although there are  few discussions on SSP for low-rank matrix recovery in the literature (cf. \cite{DF10}), we present two SSP-based results without proofs.
\begin{theorem}[Matrix SSP condition for exact recovery]\label{thm:mssp}
Let $\cA:\mathbb{R}^{n_1\times n_2}\to \mathbb{R}^m$ be a linear operator. Suppose there exists  $\Delta >0$ such that  all nonzero $\vH\in\Null(\cA)$ satisfy
$$
\frac{\|\vH\|_*}{\|\vH\|_F}\ge \sqrt{\frac{m}{\Delta}}.
$$
Assume that $\|\vX^0\|_2$ and $\alpha>0$ are fixed.  If
\beq\label{MCssd1}
m\ge   \left(2+ \frac{\|\vX^0\|_2}{\alpha}\right)^2r \Delta,
\eeq
then the null-space condition \eqref{nspMC1} holds for all $\vH\in\Null(\cA)$. Hence, \eqref{MCssd1} is sufficient for problem \eqref{N2} to recover any matrices $\vX^0$ of rank $r$ or less from measurements $\vb=\cA(\vX^0)$.
\end{theorem}

\begin{theorem}[Matrix SSP condition for stable recovery]\label{thm:MCssp} Assume that linear operator $\cA:\mathbb{R}^{n_1\times n_2}\to \mathbb{R}^m$ has the same property as it is in Theorem \ref{thm:mssp}. Let $\vX^0\in\mathbb{R}^{n_1\times n_2}$ be an \emph{arbitrary} matrix. Let $\alpha>0$ in problem \eqref{N2}. Define $C_3$ and $C_4$ in
\eqref{C34m}, which depend on $\alpha$. If
\beq\label{MCssd2}
m\ge   4\left(1+ C_3\right)^2r \Delta,
\eeq
 then the solution $\vX^*$ of \eqref{N2} satisfies
\beq\label{MCvx4}
\|\vX^*-\vX^0\|_*\le 4C_4\cdot\hat{\sigma}(\vX^0),
\eeq
where $\hat{\sigma}(\vX^0):=\sum_{i=r+1}^{\min\{n_1,n_2\}}\sigma_i(\vX^0)$ is the best rank-$r$ approximation error of $\vX^0$.
\end{theorem}

\section{Global Linear Convergence}\label{sc:glc}
Now we turn to study the numerical properties of  the linearized Bregman algorithm (LBreg) for the augmented model \eqref{P2}. In this section, we show that LBreg, as well as its two fast variants, achieves global linear convergence with \emph{no} assumptions on the solution
sparsity or aforementioned properties of matrix $\vA$. First, we review its four equivalent forms of LBreg that have appeared in different papers.
We start off with the dual gradient descent iteration \cite{Yin-LBreg-09}: give a step size $h>0$, $\vy^{(0)}=\vzero$, and $k$ starting from 0,
\begin{subequations}\label{lbreg}
\beq\label{lbregy}
\vy^{(k+1)} \gets \vy^{(k)} - h \left(-\vb + \alpha\vA\shrink(\vA^\top \vy^{(k)})\right).
\eeq
The last term of \eqref{lbregy} is the gradient of the objective function of problem \eqref{D2u}.
By letting $\vx^{(k)}:=\alpha\shrink(\vA^\top \vy^{(k)})$, one obtains the  ``primal-dual'' form \begin{align}
\label{lbregy1}
\vx^{(k+1)}&\gets \alpha\shrink(\vA^\top \vy^{(k)}),\\
\label{lbregy2}
\vy^{(k+1)}&\gets \vy^{(k)} + h (\vb - \vA\vx^{(k+1)}).
\end{align}
The same iteration is
given in \cite{YOGD08,Osher-Mao-Dong-Yin-10,Cai-Osher-Shen-lbreg2-09} as
\begin{align}
\label{lbregv1}
\vx^{(k+1)}&\gets \alpha\shrink(\vv^{(k)}),\\
\label{lbregv2}
\vv^{(k+1)}&\gets \vv^{(k)} + h \vA^\top(\vb - \vA\vx^{(k+1)}),
\end{align}
where $\vv^{(k)}=\vA^\top\vy^{(k)}$. Finally, the name ``linearized Bregman'' comes  from the iteration \cite{YOGD08}
\begin{align}
\label{lbregl1}
\vx^{(k+1)}&\gets \argmin_{\vx} D_{\ell_1}^{\vp^{(k)}}(\vx,\vx^{(k)}) + h\langle \vA^\top(\vA\vx^{(k)}-\vb),\vx\rangle + \frac{1}{2\alpha}\|\vx
- \vx^{(k)}\|_2^2,\\
\label{lbregl2}
\vp^{(k+1)}&\gets \vp^{(k)}+h\vA^\top (\vb-\vA\vx^{(k)}) - \frac{1}{\alpha}(\vx^{(k+1)}-\vx^{(k)}),
\end{align}
\end{subequations}
where $\vx^{(0)}=\vp^{(0)}=\vzero$ and
the Bregman ``distance'' $D_f^{\vp}(\cdot,\cdot)$ is defined as
$$
D_f^{\vp}(\vx,\vy) = f(\vx) - f(\vy) - \langle \vp, \vx-\vy\rangle,\quad\text{where}~ \vp \in \partial f(\vy).
$$
The last two terms of \eqref{lbregl1} replace the term $\frac{h}{2}\|\vA\vx-\vb\|_2^2$ in the original Bregman iteration.
Following \cite{YOGD08}, one can obtain \eqref{lbregv1}-\eqref{lbregv2} from \eqref{lbregl1}-\eqref{lbregl2} by setting $\vv^{(k)} =
\vp^{(k)}+h\vA^\top(\vb-\vA\vx^{(k)})+\frac{\vx^{(k)}}{\alpha}$.

It is most convenient to work with \eqref{lbregy} due to its simplicity  and gradient-descent interpretation. In the rest of this section, we let
$f(\vy)$ be the objective function of \eqref{D2u} and have $\grad f(\vy) = -\vb + \alpha\vA\shrink(\vA^\top \vy)$.


\subsection{Preliminary}
In this subsection, we prove a few key results that will be used to prove the restricted strongly convex property in the next subsection.
\begin{definition}Let $\lambda_{\min}^{++}(\vS)$ denote the minimum \emph{strictly positive} eigenvalue of a nonzero symmetric matrix $\vS$,
assuming its existence. Namely,
$$ \lambda_{\min}^{++}(\vS):=\min\{\lambda_i(\vS):\lambda_i(\vS)>0\},$$
where $\{\lambda_i(\vS)\}$ is the set of eigenvalues of $\vS$.
\end{definition}

\begin{lemma}\label{minposeig}
Let $\vA$ be a nonzero $m$-by-$n$ matrix. Let $\vD\succ\vzero$ be an $n$-by-$n$ diagonal matrix with strictly positive diagonal entries. We have
\beq\label{lmpmin}
\lambda_{\min}^{++}(\vA\vD\vA^\top)=\min_{\|\vA\alpha\|_2=1}(\vA\alpha)^\top(\vA\vD\vA^\top)(\vA\alpha).
\eeq
\end{lemma}
\begin{proof}
Let $r=\rank(\vA)\ge 1$. Since $\rank(\vA\vD\vA^\top)=r$ and $\vA\vD\vA^\top\succeq\vzero$, $\vA\vD\vA^\top$ has $r$ strictly positive eigenvalues.
Let $\lambda>0$ be a positive eigenvalue and $\vx$ be its corresponding eigenvector. Since $\vA\vD\vA^\top\vx = \lambda\vx$, we see $\vx\in\Range(\vA)$ and can thus write  $\vx = \vA\alpha_\lambda$. From this and $\rank(\vA)=r$, the eigenvectors corresponding to the $r$ strictly positive eigenvalues span $\Range(\vA)$.
Hence, 
\eqref{lmpmin} attains its minimum at the eigenvector $\vA\alpha$ corresponding to the eigenvalue
$\lambda_{\min}^{++}(\vA\vD\vA^\top)$. \end{proof}

Next, we show that a constrained eigenvalue problem, which will appear in our proof of restricted strong convexity, has a \emph
{strictly positive} minimum objective.
\begin{lemma}\label{lm:abcd} Let $\vA$ be a nonzero $m$-by-$n$ matrix, $\vB$ be an $m$-by-$\ell$ matrix, and $\vD\succ\vzero$ be a diagonal matrix of
size $n$ by $n$. Let $r:=\rank([\vA~ \vB])-\rank(\vA)$, which satisfies $0\le r\le \ell$. Let $\vc$ and $\vd$ be free vectors of sizes $n$ and $\ell$,
respectively. The  constrained eigenvalue problem
\beq
\label{vmin}
v:=\min\left\{(\vA\vc+\vB\vd)^\top(\vA\vD\vA^\top)(\vA\vc+\vB\vd):\|\vA\vc+\vB\vd\|_2=1,\vB^\top(\vA\vc+\vB\vd)\le\vzero,\vd\ge\vzero\right\}\\
\eeq
satisfies
$v\ge v_{\min}>0$, where
\beq\label{cminl}
v_{\min}:=\min_{\vC}\left\{\lambda_{\min}^{++}(\vA\vD\vA^\top+\vC\vC^\top):\vC ~\text{is an $m$-by-$p$ submatrix of}~\vB,~r\le p\le \ell\right\}.
\eeq
(If $p=0$, $\vC$ vanishes.)
\end{lemma}
Let us first explain  this lemma. If $\vA$ and $\vB$ are  orthogonal to each other (i.e., $\vA^\top \vB = \vzero$),  then $\vB^\top(\vA\vc+\vB\vd)=\vB^\top\vB\vd\le\vzero$ and $\vd\ge\vzero$ will force $\vB\vd=\vzero$ and thus reduce \eqref{vmin} to \eqref{lmpmin}. Therefore, the lemma is more about the general case where $\vA$ and $\vB$ are not orthogonal. The result \eqref{cminl} reveals that  \eqref{vmin} can go lower than \eqref{lmpmin} yet must remain strictly positive.  From another perspective, if we ignore the constraints
$\vB^\top(\vA\vc+\vB\vd)\le\vzero$ in \eqref{vmin}, then we can choose $\vc$ and $\vd\ge \vzero$ such that $\vA^\top(\vA\vc+\vB\vd)=\vzero$ and thus have $v=0$. (For example, if $r>0$, we can choose any $\vd\ge\vzero$ so that $\vB\vd\not\in \Range(\vA)$ and then choose $\vc$ so that $-\vA\vc$ equals $\vB\vd$'s projection on $\Range(\vA)$; if $r=0$, the case is trivial.) Therefore, the three constraints in \eqref{vmin} prevent $\vA^\top(\vA\vc+\vB\vd)$ from being $\vzero$. Those constraints will arise during the study of  certain KKT systems.
\begin{proof}[Proof of Lemma \ref{lm:abcd}]
Let $\vB=[\vb_1~\vb_2\cdots\vb_{\ell}]$. If $r=0$, then $\rank([\vA~ \vB])=\rank(\vA)$ and thus $\vA\vc+\vB\vd \in\Range(\vA)$. Since dropping the
constraints $\vB^\top(\vA\vc+\vB\vd)\le\vzero$ and $\vd\ge\vzero$ from \eqref{vmin} does not increase its optimal objective, we have
$v\ge\lambda_{\min}^{++}(\vA\vD\vA^\top)\ge v_{\min}>0$ from Lemma \ref{minposeig}.

Now we consider the nontrivial case $r>0$, i.e., $\Range([\vA~ \vB])\supsetneq\Range(\vA)$. Ignoring the constraints
$\vB^\top(\vA\vc+\vB\vd)\le\vzero$, we can choose $\vc$ and $\vd\ge \vzero$ such that $\vA^\top(\vA\vc+\vB\vd)=\vzero$ and thus $v=0$. (See the discussions before the proof for  example.)
Therefore, the rest of the proof focuses on the role of these constraints.


The proof is based on induction. We will show  later that as long as $\Range([\vA~ \vB])\supsetneq\Range(\vA)$,  any minimizer $(\vc^*,\vd^*)$ of
\eqref{vmin} makes at least one of the constraints
$\vB^\top(\vA\vc+\vB\vd)\le\vzero$ \emph{active}. (Minimizer $(\vc^*,\vd^*)$ exists for the following reason. Let $\vs=(\vA\vc)$ and $\vt=(\vB\vd)$ be the optimization variables instead of $\vc$ and $\vd$; then constraints $\vd\ge \vzero$ translate to $\vt\in \{\vB\vd:\vd\ge\vzero\}$, which is a closed set. Since problem \eqref{vmin} has a compact, nonempty feasible set and a continuous objective function in terms of $\vs$ and $\vt$,  there exist minimizer $(\vs^*,\vt^*)$ and thus $(\vc^*,\vd^*)$.)  Without loss of generality, suppose this active constraint is $\vb_1^\top(\vA\vc^*+\vB\vd^*)=0$.
From this, we obtain
$$v =  (\vA\vc^*+\vB\vd^*)^\top(\vA\vD\vA^\top)(\vA\vc^*+\vB\vd^*)
= (\vA\vc^*+\vB\vd^*)^\top(\vA\vD\vA^\top+\vb_1\vb_1^\top)(\vA\vc^*+\vB\vd^*).$$
We move $\vb_1$ ``from $\vB$ to $\vA$'' by introducing new matrices  $\vA_1:=[\vA~ \vb_1]$, $\vB_1:=[\vb_2~\vb_3\cdots \vb_{\ell}]$. Introduce
$$
\vD_1:=\begin{bmatrix}
\vD & \vzero\\
\vzero & 1
\end{bmatrix}
$$
so  $(\vA\vD\vA^\top+\vb_1\vb_1^\top)= (\vA_1\vD_1\vA_1^\top)$. Furthermore, drop the constraints  $\vb_1^\top(\vA\vc^*+\vB\vd^*)\le 0$ and $d_1\ge0$,
and consider the resulting problem
\beq
\label{v1min}
v_{1}:=\min_{\vc_1,\vd_1}\left\{(\vA_1\vc_1+\vB_1\vd_1)^\top(\vA_1\vD_1\vA_1^\top)(\vA_1\vc_1+\vB_1\vd_1):
\begin{array}{l}
\|\vA_1\vc_1+\vB_1\vd_1\|_2=1,\\
\vB_1^\top(\vA_1\vc_1+\vB_1\vd_1)\le\vzero,\vd_1\ge\vzero
\end{array}\right\}.
\eeq
\eqref{v1min} would have the same objective value as \eqref{vmin} if the active constraint $\vb_1^\top(\vA\vc^*+\vB\vd^*)=\vzero$ was present. As \eqref{v1min} does not have this constraint, we conclude
\beq\label{vv1}
v\ge v_1.
\eeq
We  apply the same argument  to \eqref{v1min} and then inductively to the subsequent problems: let
\beq
\label{vimin}
v_{j}:=\min_{\vc_j,\vd_j}\left\{(\vA_j\vc_j+\vB_j\vd_j)^\top(\vA_j\vD_j\vA_j^\top)(\vA_j\vc_j+\vB_j\vd_j):
\begin{array}{l}
\|\vA_j\vc_j+\vB_j\vd_j\|_2=1,\\
\vB_j^\top(\vA_j\vc_j+\vB_j\vd_j)\le\vzero,\vd_j\ge\vzero
\end{array}\right\}.
\eeq
where each $\vA_j=[\vA_{j-1}~\vb_{j}]$, $\vB_j=[\vb_{j+1}\cdots\vb_{\ell}]$, and $\vD_j=\begin{bmatrix}\vD_{j-1}&\vzero\\ \vzero& 1\end{bmatrix}$, for
$j=2,3,\ldots, p$ until either $p=\ell$
(i.e., ``all $\vb_i$'s have been moved out of $\vB$'') or $\Range([\vA_p~ \vB_p])=\Range(\vA_p)$ (i.e., the condition for the induction breaks down
when $j$ reaches $p$). The former case occurs  if $r=\ell$, and in this case, we obtain empty $\vB_{\ell}$ and $\vd_{\ell}$ and thus
$$
v_{\ell}=\min_{\vc^{\ell}}\left\{(\vA_{\ell}\vc_{\ell})^\top(\vA_{\ell}\vD_{\ell}\vA_{\ell}^\top)(\vA_{\ell}\vc_{\ell}):\|
\vA_{\ell}\vc_{\ell}\|_2=1\right\}.$$
and from the induction,
$$v\ge v_1\ge \cdots\ge v_{\ell}.$$
From $\vA_{\ell}\vD_{\ell}\vA_{\ell}^\top=\vA\vD\vA^\top+\vB\vB^\top$ and Lemma \ref{minposeig}, it follows
$$v_\ell=\lambda_{\min}^{++}(\vA\vD\vA^\top+\vB\vB^\top).$$
The latter case (i.e., $j=p<\ell$) occurs if $0<r<\ell$. In this case, $p\ge r$ and the induction gives $v\ge v_1\ge \cdots\ge v_{p}$. From
$\Range([\vA_p~ \vB_p])=\Range(\vA_p)$ and the same argument at the beginning of this proof, we have
$v_p\ge \lambda_{\min}^{++}(\vA_p\vD_p\vA_p^\top)$. By the definition of $v_{\min}$, we have $\lambda_{\min}^{++}(\vA_p\vD_p\vA_p^\top)\ge v_{\min}$
and thus $v\ge v_{\min}>0$.

Hence, Lemma \ref{lm:abcd} is proved for all three cases: $r=0$, $0<r<\ell$, and $r=\ell$.


Finally, we establish the existence of an active constraint by showing that if $\Range([\vA~ \vB])\supsetneq\Range(\vA)$, every solution of the
problem  obtained by removing the constraints $\vB^\top(\vA\vc+\vB\vd)\le\vzero$ from \eqref{vmin}, namely, 
\beq
\label{vamin}
\min_{\vc,\vd}\left\{(\vA\vc+\vB\vd)^\top(\vA\vD\vA^\top)(\vA\vc+\vB\vd):\|\vA\vc+\vB\vd\|_2=1,\vd\ge\vzero\right\},
\eeq
will violate $\vB^\top(\vA\vc+\vB\vd)\le\vzero$. Since $\Range([\vA~ \vB])\supsetneq\Range(\vA)$, as been argued above, one can choose $\vc$ and $\vd\ge\vzero$ such that
$\vA\vc+\vB\vd\in\Null(\vA)$ and thus $(\vA\vc+\vB\vd)^\top(\vA\vD\vA^\top)(\vA\vc+\vB\vd)=0$. (See the discussions before the proof for  example.) Therefore, any solution $(\bar{\vc},\bar{\vd})$ of
\eqref{vamin}  must attain the 0 objective, so
\beq\label{aac}
\vA^\top(\vA\bar{\vc}+\vB\bar{\vd})=\vzero.
\eeq
Suppose
\beq\label{bac}
\vB^\top(\vA\bar{\vc}+\vB\bar{\vd})\le\vzero.
\eeq
 i.e., no  constraint is violated. Then, from $\bar{\vd}\ge \vzero$, \eqref{bac}, and \eqref{aac}, it follows
\begin{align}
\bar{\vd}^\top \vB^\top(\vA\bar{\vc}+\vB\bar{\vd}) \le &\,\vzero,\\
\bar{\vc}^\top \vA^\top(\vA\bar{\vc}+\vB\bar{\vd}) = &\, \vzero,
\end{align}
so
$$
\|\vA\bar{\vc}+\vB\bar{\vd}\|_2^2 =\, \bar{\vd}^\top \vB^\top(\vA\bar{\vc}+\vB\bar{\vd}) + \bar{\vc}^\top \vA^\top(\vA\bar{\vc}+\vB\bar{\vd}) \\
\le \vzero,
$$
which contradicts the constraint $\|\vA\vc+\vB\vd\|_2=1$. Therefore, $\vB^\top(\vA\bar{\vc}+\vB\bar{\vd})\le\vzero$ cannot hold, and at least one of
these constraints must be violated. Clearly, this argument applies to problem \eqref{vimin} for $j=1,2,\ldots,$ as long as $j\le\ell$ and
$\Range([\vA_j~ \vB_j])\supsetneq\Range(\vA_j)$.
\end{proof}
\begin{lemma}\label{lm:shrkiq}
Let $\shrink$ be the shrinkage operator 
$\shrink(s)= \hbox{sign}(s) \max\{|s|-1,0\}$. Then the following inequality
\beq\label{srkineq}
(s-s^*)\cdot(\shrink(s)-\shrink(s^*)) \ge \frac{|\shrink(s^*)|}{|\shrink(s^*)|+2} \cdot (s-s^*)^2\ge 0
\eeq
holds for  $\forall s,s^*\in\mathbb{R}$. The first equality holds when $s=-\sign(s^*)$.
\end{lemma}
\begin{proof}
The first inequality in \eqref{srkineq} can be proved by elementary case-by-case analysis. The second one is trivial.
\end{proof}

\subsection{Globally Linear Convergence}
In this subsection, we show that the  LBreg iteration \eqref{lbregy},
as a fixed-step size gradient descent iteration for \eqref{D2u},
 generates a globally linearly convergent sequences $\{\vy^k\}$ and $\{\vx^k\}$.

To do this, we need the following theorem from \cite{Yin-LBreg-09} with our modifications for better clarity. Below, we use the notion
$$\shrink(\vz):=\shrink_1(\vz)=\vz - \Proj_{[-1,1]^n}(\vz) = \sign(\vz)\max\{|\vz|-\vone,\vzero\},$$
where $\sign(\cdot)$, $|\cdot|$, and $\max\{\cdot,\cdot\}$ are component--wise operations.
\begin{theorem}
Let $f$ denote the objective function of problem \eqref{D2u}, and $\vx^*$ denote the solution of \eqref{P2}, which is unique since it has a strictly
convex objective. Define coordinate sets $\cS_+, \cS_-, \cS_0$ as the sets of positive, negative, and zero components of $\vx^*$, respectively.
Corresponding to $\cS_+, \cS_-, \cS_0$,  decompose
\beqs
\vA & = & [\vA_+,\vA_-,\vA_0],\\
\vx^* & = & [\vx^*_+;\vx^*_-;\vx^*_0] .
\eeqs
Then, the set of solutions of \eqref{D2u} is given by
\begin{subequations}\label{defY}
\beqn\label{defy1}
\cY^* & = & \{\vy'\in\mathbb{R}^m:\alpha\shrink(\vA^\top \vy') = \vx^*\}\\
\label{defyset}
& = & \{\vy'\in\mathbb{R}^m:\vA_+^\top \vy' -\vone= \alpha^{-1}\vx_+^*,~\vA_-^\top \vy' +\vone= \alpha^{-1}\vx_-^*,~-\vone \le \vA_0^\top \vy'\le \vone\},
\eeqn
\end{subequations}
which is a convex set.
Furthermore, $\grad f(\vy')=\vzero,~\forall \vy'\in\cY^*$. 
\end{theorem}
\begin{proof}Any $\vy'\in \cY^*$ must satisfy 
the strong duality condition, namely, the primal objective equal to the dual objective: $-f(\vy') = \|\vx^*\|_1+\frac{1}{2\alpha}\|\vx^*\|_2^2$. From this and $\vA\vx^*=\vb$, it is easy to derive $\alpha\shrink(\vA^\top \vy') = \vx^*$ using a case-by-case analysis on the sign of $x^*_i$.
Conversely, since $\nabla f(\vy) = -\vb + \vA (\alpha\shrink(\vA^\top \vy))$ and $\vA\vx^*=\vb$,  any $\vy'$ obeying $\alpha\shrink(\vA^\top \vy') = \vx^*$\
satisfies $\nabla f(\vy')=\vzero$. Then, $\vy'\in \cY^*$.

By the definition \eqref{defyset}, $\cY^*$ is a polyhedron, so it is convex.
\end{proof}

In general, the two sets
of equality equations in \eqref{defyset} do not define a unique $\vy^*$, so $\cY^*$ can include multiple solutions. 

A typical tool for obtaining  global convergence at a linear rate (or, global geometric convergence) is the strong convexity of the objective
function. A function  $g$ is strongly convex with a constant $c$ if it satisfies
\beq
\label{strcvx}
\langle\vy - \vy', \grad f(\vy) - \grad f(\vy')\rangle \ge c\|\vy- \vy'\|^2,\quad \forall \vy,\vy'\in\dom\,f.
\eeq
Strong convexity, however, does not hold for our $f(\vy)$ since $\grad f(\vy^*)=\vzero,~\forall \vy^*\in\cY^*$, while $\cY^*$ is not necessarily a singleton.
Nevertheless, we  establish  the ``restricted'' strong convexity \eqref{rescvx} below.
\begin{lemma}[Restricted strong convexity]
\label{lem6}
Consider problem 
\eqref{D2u} with a nonzero $m$-by-$n$ matrix $\vA$ and nonzero vector $\vb$. Assume that  $\vA\vx=\vb$ are consistent.  Let $\Proj_{\cY^*}(\vy)$ denote the Euclidean projection of $\vy$ to the solution set $\cY^*$.  The
objective function $f$ of \eqref{D2u} satisfies
\beq
\label{rescvx}
\langle \vy - \Proj_{\cY^*}(\vy),\grad f(\vy)\rangle \ge \nu \|\vy - \Proj_{\cY^*}(\vy)\|^2, \quad \forall\, \vy,
\eeq
where constant
\beq\label{constnu}
\nu=\lambda_{\vA}\cdot \left(\min_{i\in \supp(\vx^*)}\frac{\alpha|x^*_i|}{|x^*_i|+2\alpha}\right)>0,
\eeq
and $\lambda_{\vA}=\min\left\{\lambda^{++}_{\min}({\vC}{\vC}^\top):{\vC}~\text{is a nonzero submatrix of}~\vA~\text{of $m$ rows}\right\}$.

\end{lemma}
Note that if we let $\vy' =\Proj_{\cY^*}(\vy)$ and from $\grad f(\vy') = \vzero$, \eqref{rescvx} becomes $\langle \vy - \vy',\grad f(\vy)-\grad f(\vy')\rangle \ge \nu \|\vy - \vy'\|^2$. Hence, \eqref{rescvx} is the restriction of \eqref{strcvx} to the specially chosen $\vy'$. Yet, this will be enough for global linear convergence.
\begin{proof}[Proof of Lemma \ref{lem6}] Since  $\vA\vx=\vb$ are consistent, problem \eqref{P2} has a unique solution $\vx^*$, so $\cY^*$ is well-defined and nonempty. If $\vy\in\cY^*$, then $\vy=\Proj_{\cY^*}(\vy)$ and thus  \eqref{rescvx} holds trivially. To show \eqref{rescvx} for $\vy\not\in\cY^*$, we shall consider
\beq\label{ratio}
\min \left\{\frac{\langle \vy - \vy',\grad f(\vy)-\grad f(\vy')\rangle}{\langle \vy -\vy',\vy -\vy'\rangle}:\vy-\vy'\not=0,~\vy'=\Proj_{\cY^*}(\vy).\right\}
\eeq
%
%
The proof is divided to three parts. The first part works out $\vy'=\Proj_{\cY^*}(\vy)$ and express $\vy - \vy'$ in terms of submatrices of $\vA$. The second part establishes $\langle \vy - \vy',\grad f(\vy)-\grad f(\vy')\ge (\vy-\vy')^\top \vM(\vy-\vy')$, where $\vM\succeq \vzero$ also depends on  submatrices of $\vA$. The last part invokes Lemma \ref{lm:abcd} to obtain a strictly positive lower bound for \eqref{ratio}. Most of the effort is to  decompose $\vA$ into the submatrices and understand how they contribute to $ \vy - \vy' $ and $\grad f(\vy)-\grad f(\vy')$.

Part 1. By definition, $\vy'=\Proj_{\cY^*}(\vy)$ is the solution of
\begin{equation}\label{prj}
\min_{\bar{\vy}}\left\{\frac{1}{2}\|{\bar{\vy}} - \vy\|_2^2:~\bar{\vy}\in \cY^*\right\}. \end{equation}
Hence, $\vy'$ satisfies the KKT conditions of \eqref{prj}. Using the expression of $\cY^*$ in \eqref{defyset}, these conditions are
\begin{subequations}\label{kkt}
\begin{eqnarray}
\label{pd}
\vy-\vy' & = & \vA_+\lambda_+ + \vA_-\lambda_-+\vA_0(\vu-\boldsymbol{\ell}),\\
\vy' & \in & \cY^*,\\
\label{lu}
\boldsymbol{\ell},\vu & \ge & \vzero,\\
\label{comp}
(\vone-\vA_0^\top \vy')^\top  \vu + (\vone+\vA_0^\top \vy')^\top  \boldsymbol{\ell} & = & \vzero,
\end{eqnarray}
\end{subequations}
where $\lambda_+$ and $\lambda_-$ are the Lagrange multipliers for the two equality  conditions in \eqref{defyset} and $\boldsymbol{\ell}$ and $\vu$ are those for the first and second inequality conditions in \eqref{defyset}, respectively. Equation \eqref{comp} is the so-called complementarity condition, which together with \eqref{lu}, gives the following three cases for $\forall i\in \cS_0$:
\beq\label{compc}
\ell_i=0,~ u_i=0;\quad \text{if}~u_i>0, ~ \text{then} ~ \vA_i^\top \vy' = 1,~\ell_i=0;\quad\text{if}~ \ell_i>0, ~ \text{then} ~ \vA_i^\top \vy' = -1,~u_i=0.
\eeq
Part 2. Let
$
\vA_\pm = [\vA_+, \vA_-].
$
We  first argue that $\vA_{\pm}$ is a nonzero submatrix of $\vA$. Since $\vA$ and $\vb$ are both nonzero, the solution $\vx^*$ to problem \eqref{P2} is nonzero. If some column $\va_i$ of $\vA$ is a zero vector, then $x_i$ is free from the constraints $\vA\vx=\vb$ and thus $x^*_i=0$. Hence, all the columns of $\vA_{\pm}$ are nonzero vectors.

From $\grad f(\vy)=-\vb+\alpha \vA\shrink(\vA^\top  \vy)$ and $\vzero=\nabla f(\vy')=-\vb+\alpha \vA\shrink(\vA^\top  \vy')$, we obtain
\begin{subequations}\label{parts}
\begin{align}
\langle\vy - \vy', \nabla f(\vy)\rangle  = \hspace{0pt}\langle\vy - \vy', \nabla f(\vy)-\nabla f(\vy')\rangle
= &\, \alpha \langle \vA^\top  \vy - \vA^\top  \vy', \shrink(\vA^\top  \vy) - \shrink(\vA^\top  \vy')\rangle\\
\label{pmpart}
 =&\,  \alpha \langle \vA_{\pm}^\top  \vy - \vA_{\pm}^\top  \vy', \shrink(\vA_{\pm}^\top  \vy) - \shrink(\vA_{\pm}^\top  \vy')\rangle \\
\label{zopart}
& +\alpha \langle \vA_{0}^\top  \vy - \vA_{0}^\top  \vy', \shrink(\vA_{0}^\top  \vy) - \shrink(\vA_{0}^\top  \vy')\rangle.
\end{align}
\end{subequations}
By definition, every component of $ \shrink(\vA_{\pm}^\top  \vy')=\alpha^{-1}\vx^*_\pm$ is nonzero, and all components of $
\shrink(\vA_{0}^\top  \vy')=\alpha^{-1}\vx_0^*$ are zero. For this reason, we deal with \eqref{pmpart} and \eqref{zopart} separately.

Applying inequality \eqref{srkineq} to  \eqref{pmpart}, we can ``remove'' the ``$\shrink$'' operators for it as
\begin{eqnarray}
\nonumber
 \alpha\langle \vA_{\pm}^\top  \vy - \vA_{\pm}^\top  \vy^*, \shrink(\vA_{\pm}^\top  \vy) - \shrink(\vA_{\pm}^\top
\vy')\rangle & = & \alpha\sum_{i\in \cS_\pm} (\va^\top_i \vy - \va^\top_i \vy')\cdot (\shrink(\va^\top_i \vy) - \shrink(\va^\top_i \vy'))\\\
\nonumber
& \ge & \alpha\sum_{i\in \cS_\pm} \frac{\alpha^{-1}|x^*_i|}{\alpha^{-1}|x^*_i|+2}\cdot( \va^\top_i \vy - \va^\top_i \vy')^2\\
\label{pmpart1}
\quad& = & \alpha(\vy-\vy')^\top  \vA_{\pm}\hat{\vD} \vA_{\pm}^\top  (\vy-\vy'),
\end{eqnarray}
where $\hat{\vD} := \diag\left(\frac{\alpha^{-1}|x^*_i|}{\alpha^{-1}|x^*_i|+2}\right)_{i\in \supp(\vx^*)}\succ 0$.
Equation \eqref{pmpart1} along is not enough to bound \eqref{ratio} from zero since $\vA_{\pm}$ can have more columns than rows and $  \vA_{\pm}\hat{\vD} \vA_{\pm}^\top $ can be rank deficient. So, we need to include \eqref{zopart} in the analysis, and we begin with a decomposition of the involved matrix $\vA_0$:
$$\vA_0=[\vA_1 ~\vA_2 ~\vA_3~\vA_4 ~\vA_5]$$
according to the criteria
\begin{subequations}\label{see}
\begin{align}
\label{se0}\vy-\vy' & = \vA_\pm\lambda_\pm +\vA_1\vu_1+\vA_2\vu_2-\vA_3\boldsymbol{\ell}_3 - \vA_4\boldsymbol{\ell}_4,~\text{where}~\vu_1,\vu_2,\boldsymbol{\ell}_3,\boldsymbol{\ell}_4>\vzero,\\
\label{se1}\vA_1^\top  \vy &> +\vone, \\
\label{se2}\vA_2^\top  \vy &\le +\vone, \\
\label{se3}\vA_3^\top  \vy &< -\vone, \\
\label{se4}\vA_4^\top  \vy &\ge -\vone.
\end{align}
\end{subequations}
Equations \eqref{see} mean the followings: (i) 
the projected point $\vy'$ is actively confined by  the boundaries of $\cY^*$ involving $[\vA_1~ \vA_2~ \vA_3~ \vA_4]$ (c.f., the last term of \eqref{defyset}); (ii) $\vA_5$ does not contribute to $\vy-\vy'$; (iii) by applying  \eqref{compc} and \eqref{se1}--\eqref{se4}, we get  $\vA_1^\top  \vy'=\vone$, $\vA_3^\top  \vy'=-\vone$ and can thus simplify the components of \eqref{zopart} involving $\vA_1$ and $\vA_3$ as follows:
\begin{subequations}
\begin{align}
\shrink(\vA_1^\top  \vy)-\shrink(\vA_1^\top  \vy')=&\shrink(\vA_1^\top  \vy)=\vA_1^\top  \vy-\vone=\vA_1^\top \vy-\vA_1^\top \vy',\\
\shrink(\vA_3^\top  \vy)-\shrink(\vA_3^\top  \vy')=&\shrink(\vA_3^\top  \vy)=\vA_3^\top \vy+\vone=\vA_3^\top  \vy-\vA_3^\top  \vy'.
\end{align}
\end{subequations}
Now we ``drop'' the components of \eqref{zopart} involving $\vA_2$, $\vA_4$, and $\vA_5$ as follows: from \eqref{srkineq}, it follows that $\langle \vA_{i}^\top  \vy - \vA_{i}^\top  \vy', \shrink(\vA_{i}^\top  \vy) - \shrink(\vA_{i}^\top  \vy')\rangle\ge 0$ for $i=2,4,5$.
Hence,
\begin{align}
\nonumber
 \alpha\langle \vA_{0}^\top  \vy - \vA_{0}^\top  \vy', \shrink(\vA_{0}^\top  \vy) - \shrink(\vA_{0}^\top  \vy')\rangle = &\alpha \sum_{i=1}^5 \langle \vA_{i}^\top  \vy - \vA_{i}^\top  \vy', \shrink(\vA_{i}^\top  \vy) - \shrink(\vA_{i}^\top  \vy')\rangle\\
\nonumber
\ge &\alpha \sum_{i=1,3} \langle \vA_{i}^\top  \vy - \vA_{i}^\top  \vy', \shrink(\vA_{i}^\top  \vy) - \shrink(\vA_{i}^\top  \vy')\rangle\\
\label{opart1}
= &\alpha (\vy - \vy')^\top (\vA_1\vA_1^\top  + \vA_3\vA_3^\top )(\vy-\vy').
\end{align}
Now we combine \eqref{pmpart1} and \eqref{opart1}.
 Define $\bar{\vA}=[\vA_{\pm}~\vA_1~ (-\vA_3)]$,
$\bar{\vc}=[\lambda_{\pm};\vu_1;\boldsymbol{\ell}_3]$, $\bar{\vB}=[\vA_2~ (-\vA_4)]$,  $\bar{\vd}=[\vu_2;\boldsymbol{\ell}_4]$, and $\bar{\vD}=\begin{bmatrix}\hat{\vD} & \vzero\\
\vzero & I\end{bmatrix}$. By \eqref{se0}, we have $\vy-\vy'=\bar{\vA}\bar{\vc} + \bar{\vB}\bar{\vd}$ and $\bar{\vd}\ge\vzero$.
Plugging \eqref{pmpart1} and \eqref{opart1} into \eqref{parts}, we get
\beq
\label{lvy}
\langle\vy - \vy', \nabla f(\vy)\rangle\ge \alpha(\bar{\vA}\bar{\vc} + \bar{\vB}\bar{\vd})^\top (\bar{\vA}\bar{\vD}\bar{\vA}^\top)(\bar{\vA}\bar{\vc} + \bar{\vB}\bar{\vd})
\eeq
However, \eqref{lvy} is still not enough to bound \eqref{ratio} from zero since $\bar{\vA}\bar{\vD}\bar{\vA}^\top$ may still be rank deficient.

Part 3. To bound \eqref{ratio}, we now include the ``dropped'' parts of $\vA$  and apply Lemma \ref{lm:abcd}.
From \eqref{compc}, we have  $\vA_2^\top  \vy'=\vone$ and $\vA_4^\top  \vy'=-\vone$, and further from \eqref{se2} and \eqref{se4},
\begin{align*}
\vone\ge\vA_2^\top {\vy}=&\vA_2^\top \vy'+\vA_2^\top (\vy-\vy')=+\vone + \vA_2^\top (\bar{\vA}\bar{\vc} + \bar{\vB}\bar{\vd}),\\
-\vone\le\vA_4^\top {\vy}=&\vA_4^\top \vy'+\vA_4^\top (\vy-\vy')=-\vone + \vA_4^\top (\bar{\vA}\bar{\vc} + \bar{\vB}\bar{\vd}),
\end{align*}
or written compactly,
\beq\label{bcond}
\bar{\vB}^\top (\bar{\vA}\bar{\vc} + \bar{\vB}\bar{\vd})\le \vzero.
\eeq
Now for the objective of \eqref{ratio}, we apply \eqref{lvy} and then  Lemma~\ref{lm:abcd} to obtain
\begin{align*}
\frac{\langle\vy - \vy', \nabla f(\vy)\rangle}{\langle\vy - \vy',\vy - \vy'\rangle}\ge\, &\alpha\cdot\min\left\{\frac{(\bar{\vA}\bar{\vc} + \bar{\vB}\bar{\vd})^\top (\bar{\vA}\bar{\vD}\bar{\vA}^\top)(\bar{\vA}\bar{\vc} + \bar{\vB}\bar{\vd})}{(\bar{\vA}\bar{\vc} + \bar{\vB}\bar{\vd})^\top(\bar{\vA}\bar{\vc} + \bar{\vB}\bar{\vd})}:\bar{\vA}\bar{\vc} + \bar{\vB}\bar{\vd}\not=\vzero,\bar{\vd}\ge\vzero, \bar{\vB}^\top (\bar{\vA}\bar{\vc} + \bar{\vB}\bar{\vd})\le \vzero
\right\}\\
\ge\,&\alpha\cdot\min\{\lambda_{\min}^{++}(\bar{\vA}\bar{\vD}\bar{\vA}^\top+\bar{\vC}\bar{\vC}^\top):\bar{\vC}~\text{is an $m$-by-$p$ submatrix of}~\bar{\vB},~ p\ge 0 \}
\end{align*}
Note that under our convention,  an $m$-by-$0$ matrix  vanishes. Since matrix $\bar{\vA}$ contains the nonzero matrix $\vA_{\pm}$ as a submatrix, $\bar{\vA}\bar{\vD}\bar{\vA}^\top+\bar{\vC}\bar{\vC}^\top$ is nonzero. Therefore, we have
\begin{align*}
\frac{\langle\vy - \vy', \nabla f(\vy)\rangle}{\langle\vy - \vy',\vy - \vy'\rangle}\ge\,&\alpha\cdot (\min_{i}(\bar{\vD})_{ii})\cdot\underbrace{\min\{\lambda^{++}_{\min}({\vC}{\vC}^\top):{\vC}~\text{is a nonzero submatrix of}~\vA~\text{of $m$ rows}\}}_{\lambda_\vA}\\
\ge\,&\left(\min_{i\in \supp(\vx^*)}\frac{\alpha|x^*_i|}{|x^*_i|+2\alpha}\right)\cdot\lambda_\vA\\
=\,&\nu.
\end{align*}

\end{proof}
\begin{remark}
If the entries of $\vA$ are in general positions, i.e., any $m$ distinct columns of $\vA$ are linearly independent, or in other words, $\vA$ has
completely full rank  \cite{LW10},
then all $m$-by-$m$ submatrices of $\vA$ have full rank and thus $\lambda_{\vA}=\{\lambda_{\min}(\vC\vC^\top):\vC~\text{is an $m$-by-$m$ submatrix of}~\vA\}.$ This is often the case when those entries are samples from i.i.d. subgaussian distributions, or the columns of $\vA$ are data vector independent of one another. In general, the submatrix $\vC^*$ achieving the minimum $\lambda_{\vA}$ has the maximum number of independent columns, i.e., it contains $r$ columns from $\vA$ where $r=\rank(\vA)$.
\end{remark}


With the restricted strong convexity property, we next show the main convergence result with the help of the  standard notion of  point--to--set distance
$$
\dist(\vz,\cZ):=\min_{\vz'}\{\|\vz-\vz'\|_2:\vz'\in\cZ\},
$$
where $\vz$ is a vector and $\cZ$ is a set of vectors. By convention,  the convergence $\dist(\vz^k,\cZ)\to 0$ is called globally Q-linear if there exists $\mu\in(0,1)$ such that $\dist(\vz^{k+1},\cZ)/\dist(\vz^k,\cZ)\le \mu$ for all $k$, and the convergence $s^k\to 0$ is called globally R-linear if there exists a globally Q-linear converging sequence $t^k\to 0$ such that $|s^k|\le |t^k|$. Unlike Q-linear convergence, R-linear convergence does not require $|s^k|$ to be monotonic in $k$.
\begin{theorem}
\label{mainconv}
Consider problem 
\eqref{D2u} with a nonzero $m$-by-$n$ matrix $\vA$ and nonzero vector $\vb$. Assume that  $\vA\vx=\vb$ are consistent.
Let $f$ be the objective function of problem \eqref{D2u} and $f^*$ be the optimal objective value.
The  linearized Bregman iteration \eqref{lbregy} starting from any $\vy^{(0)}\in\mathbb{R}^m$ with step size $$0<h< 2\nu/(\alpha^2 \|\vA\|^4),$$
 where the strong convexity constant $\nu$ is given in \eqref{constnu}, generates a globally Q-linearly converging
sequence $\{\vy^{(k)}, k\ge 1\}$
\beq\label{dyk}
\dist(\vy^{(k)},\cY^*)\le \, \left(1- 2h\nu + h^2 \alpha^2 \|\vA\|_2^4\right)^{k/2}  \dist(\vy^{(0)},\cY^*),
\eeq
where $\cY^*$ is given in \eqref{defY}. The objective value sequence converges R-linearly as
\beq
\label{dyk2}
f(\vy^{(k)}) - f^* \le \,\frac{L}{2}\left(1-2h\nu + h^2\alpha^2 \|\vA\|_2^4 \right)^{k}\dist^2(\vy^{(0)},\cY^*).
\eeq
Furthermore, $\{\vx^{(k)}\}$ is a globally R-linear converging sequence since
\beq\label{dyk3}
\|\vx^{(k+1)}-\vx^{*}\|_2\le \alpha \|\vA\|_2\cdot\dist(\vy^{(k)},\cY^*).
\eeq
\end{theorem}
\begin{proof}
For each $k$, let $\vy'^{(k)}:=\Proj_{\cY^*}(\vy^{(k)})$. Hence, $\dist(\vy^{(k)},\cY^*)=\|\vy^{(k)}-\vy'^{(k)}\|_2$.
Using this projection property, we have
\begin{subequations}
\begin{align}
\|\vy^{(k+1)}-\vy'^{(k+1)}\|^2_2 \le&\, \|\vy^{(k+1)}-\vy'^{(k)}\|^2_2\\
=&\, \|\vy^{(k)}-\vy'^{(k)}-h\nabla f(\vy^{(k)})\|_2^2\\
=&\, \|\vy^{(k)}-\vy'^{(k)}\|_2^2-2h\langle\nabla f(\vy^{(k)}),\vy^{(k)}-\vy'^{(k)}\rangle+h^2\|\nabla f(\vy^{(k)})-\nabla f(\vy'^{(k)})\|_2^2\\
\label{hhvl}
\le &\left(1-2h\nu\right) \|\vy^{(k)}-\vy'^{(k)}\|_2^2 + h^2\|\alpha\vA \shrink(\vA^\top \vy^{(k)})-\alpha\vA \shrink(\vA^\top \vy'^{(k)})\|_2^2\\
\label{hhvl1}
\le & \left(1-2h\nu\right) \|\vy^{(k)}-\vy'^{(k)}\|_2^2 + h^2\alpha^2\|\vA\|_2^2 \|\vA^\top \vy^{(k)} - \vA^\top \vy'^{(k)}\|_2^2\\
\le &  \left(1-2h\nu + h^2 \alpha^2 \|\vA\|_2^4 \right) \|\vy^{(k)}-\vy'^{(k)}\|_2^2
\end{align}
\end{subequations}
where we have used the nonexpansive property of the shrinkage operator (cf. \cite{Hale-Yin-Zhang-07-theory}). Hence, we obtain \eqref{dyk}.

To get \eqref{dyk2}, we recall for any convex $f$ with $L$-Lipschitz $\nabla f$,
$f(\vy) - f(\vx) \le \langle\nabla f(\vx), \vy-\vx\rangle + \frac{L}{2}\|\vx - \vy\|_2^2$ (see Theorem 2.1.5 in \cite{Nesterov-book-04}).
Let $\vy=\vy^{(k)}$ and $\vx =\vy'^{(k)}$. We have $f(\vy'^{(k)})=f^*$, $\nabla f(\vy'^{(k)})=\vzero$ and from \eqref{dyk},
$$
f(\vy^{(k)}) - f^*\le \frac{L}{2}\|\vy^{(k)}-\vy'^{(k)}\|_2^2\le \frac{L}{2}\left(1-2h\nu + h^2\alpha^2 \|\vA\|^4\right)^{k}
\dist^2(\vy^{(0)},\cY^*), $$
which shows \eqref{dyk2}. When $0<h< 2\nu/(\alpha^2 \|\vA\|^4)$, we have  $\left(1-2h\nu + h^2\alpha^2 \|\vA\|^4\right)<1$.  Due to \eqref{lbregy1}, \eqref{defy1}, and the non-expansiveness of $\shrink(\cdot)$, we get
\begin{subequations}
\begin{align}
\|\vx^{(k+1)}-\vx^{*}\|_2&\le \|\alpha\shrink(\vA^\top \vy^{(k)})-\alpha\shrink(\vA^\top \vy'^{(k)})\|_2\\
&\le \alpha\|\vA^\top \vy^{(k)} - \vA^\top \vy'^{(k)}\|_2\\
&\le \alpha \|\vA\|_2\cdot\|\vy^{(k)} -\vy'^{(k)}\|_2,
\end{align}
\end{subequations}
which gives \eqref{dyk3}.
\end{proof}
\begin{remark}
If we set $h=\nu/(\alpha^2 \|\vA\|_2^4)$, then the geometric decay factor $\left(1- 2h\nu + h^2 \alpha^2 \|\vA\|_2^4\right)=\left(1-\nu^2/(\alpha^2 \|\vA\|_2^4)\right)$. Hence, we find the convergence rate affected by $\nu$, $\alpha$, and $\|\vA\|_2$. From the definition of $\nu$ in \eqref{constnu}, we get
\beq\label{ratt}
\text{decay factor}=1-\frac{\nu^2}{\alpha^2 \|\vA\|_2^4}=1-\omega^2\cdot\kappa^2,
\eeq
where
\begin{align*}
\omega&:=\min_{i\in\supp(\vx^*)}\frac{|x^*_i|/\alpha}{2+|x^*_i|/\alpha}\\
\kappa&:=\min\left\{\frac{\lambda^{++}_{\min}({\vC}{\vC}^\top)}{\lambda_{\max}(\vA \vA^\top)}:{\vC}~\text{is a nonzero submatrix of}~\vA~\text{of $m$ rows}\right\}.
\end{align*}
The constant $\kappa$ is similar to the ``condition number'' of $\vA$.
Let $r^*=(\max_{i\in\supp(\vx^*)}x_i^*)/(\min_{i\in\supp(\vx^*)}x_i^*)$ denote the dynamic range of $\vx^*$. If we set $\alpha = C\|\vx^*\|_\infty$, then
$$\omega = (1+2C r^*)^{-1}.$$
For recovering a sparse vector, recall that both the simulations in Section \ref{motv} and the analysis in Section \ref{rcg} show that if $\vx^*$ has faster decaying nonzero entries, $C$ can be set smaller. So, when $r^*$ is large, one can choose a small $C$ to counteract.

The proved rate of convergence is quite conservative. The dependence on the solution dynamic range is due to \eqref{pmpart1}, which considers the worst case of \eqref{srkineq}, yet when this worst case happens, the inequality between \eqref{hhvl} and \eqref{hhvl1} can be improved due to  properties of the shrinkage operator. In addition, our analysis on the global rate does not exploit the possibility that the algorithm may reach the optimal active set in a finite number of iterations and then exhibit faster linear convergence, typically at a rate depending only on the active set of columns of $\vA$ and independent of the solution's dynamic range.

The step size $h\le 2\nu/(\alpha^2 \|\vA\|_2^4)$ is also very conservative. As one will see in the simulation results in the next section, classical techniques for gradient descents such as line search can significantly accelerate the convergence.
\end{remark}

\subsection{Extensions to two faster variants of LBreg}\label{sc:ext}
We extend the linear convergence results to two variants of  LBreg (iteration \eqref{lbreg}) that can run significantly faster than LBreg: \emph{BB-line-search} \cite{Yin-LBreg-09} and \emph{kicking} \cite{Osher-Mao-Dong-Yin-10}. The former dynamically sets the step size $h$ in \eqref{lbreg} by the Barzilai-Borwein method with nonmontone line search using techniques from \cite{nonmono04}. The latter is a simple add-on to iteration \eqref{lbreg} to consolidate a sequence of consecutive iterations in which $\vx^k$ is unchanged. If  $\vx^k=\cdots=\vx^{k+j}$,  \cite{Osher-Mao-Dong-Yin-10} shows that $\vy^k,\ldots,\vy^{k+j}$ stay on the same line, so it is easy to skip all the intermediate iterations and go directly to the end of the line.

Obviously, since \emph{kicking} only skips certain LBreg iterations, it remains have global linear convergence. On the other hand, given strong convexity, Theorems 3.1 and 3.2 of \cite{nonmono04} shows that  \emph{BB-line-search} also enjoys global linear convergence (though the results are weakened to the R-linear convergence of $\vA\vx^{(k)}-\vb$ in our case); it is not difficult to verify that the proof of the theorem remains to hold given only restricted strong convexity\footnote{In \cite{nonmono04}, Theorem 3.1 relies on its inequality (3.4), which in turn require inequalities (3.3) and (3.2)  to hold between a current point and its projection to the solution set. The latter is precise our \eqref{strcvx}. Theorem 3.2 needs (3.12) and in turn (3.11). (3.11) is obtained from (3.1) restricted to between a current point and its projection to the solution set, which can be proved by assuming (3.2) or our \eqref{strcvx}.}.

\subsection{Numerical Demonstration}\label{sc:num}
We present the results of simple tests to demonstrate the convergence of three  algorithms: the original LBreg iteration \eqref{lbreg}, \emph{kicking} \cite{Osher-Mao-Dong-Yin-10}, and \emph{BB-line-search} \cite{nonmono04,Yin-LBreg-09}. Their numerical efficiency and properties  have been  previously studied in papers \cite{Osher-Mao-Dong-Yin-10,Yang-Moller-Osher-11,Huang-Ma-Goldfarb-11} and are not the focus of this paper, so we merely use two examples to illustrate  global linear convergence. We generated two compressive sensing tests where both tests had signals $\vx^0$ with 512 entries, out of which $50$ were nonzero and sampled from the standard Gaussian distribution (for Figure \ref{fig:convg}) or the Bernoulli distribution (for Figure \ref{fig:convg_flat}). Both tests had the same sensing matrix $\vA$ with 256 rows and entries sampled from the standard Gaussian distribution. We set $\alpha = 10\|\vx^0\|_\infty$ in each test and stopped all the three algorithms upon $\|\grad f(\vy)\|_2 < 10^{-6}$. The iterative errors $\|\vx^k-\vx^*\|_2$ and  $\|\vy^k-\vy^*\|_2$ of the three algorithms are depicted in Figures \ref{fig:convg} and \ref{fig:convg_flat}.
\begin{figure}[ht]
\centering
\subfigure[$\ell_2$ error of primal variable $\vx^{(k)}$]
{
\includegraphics[width=0.45\textwidth]{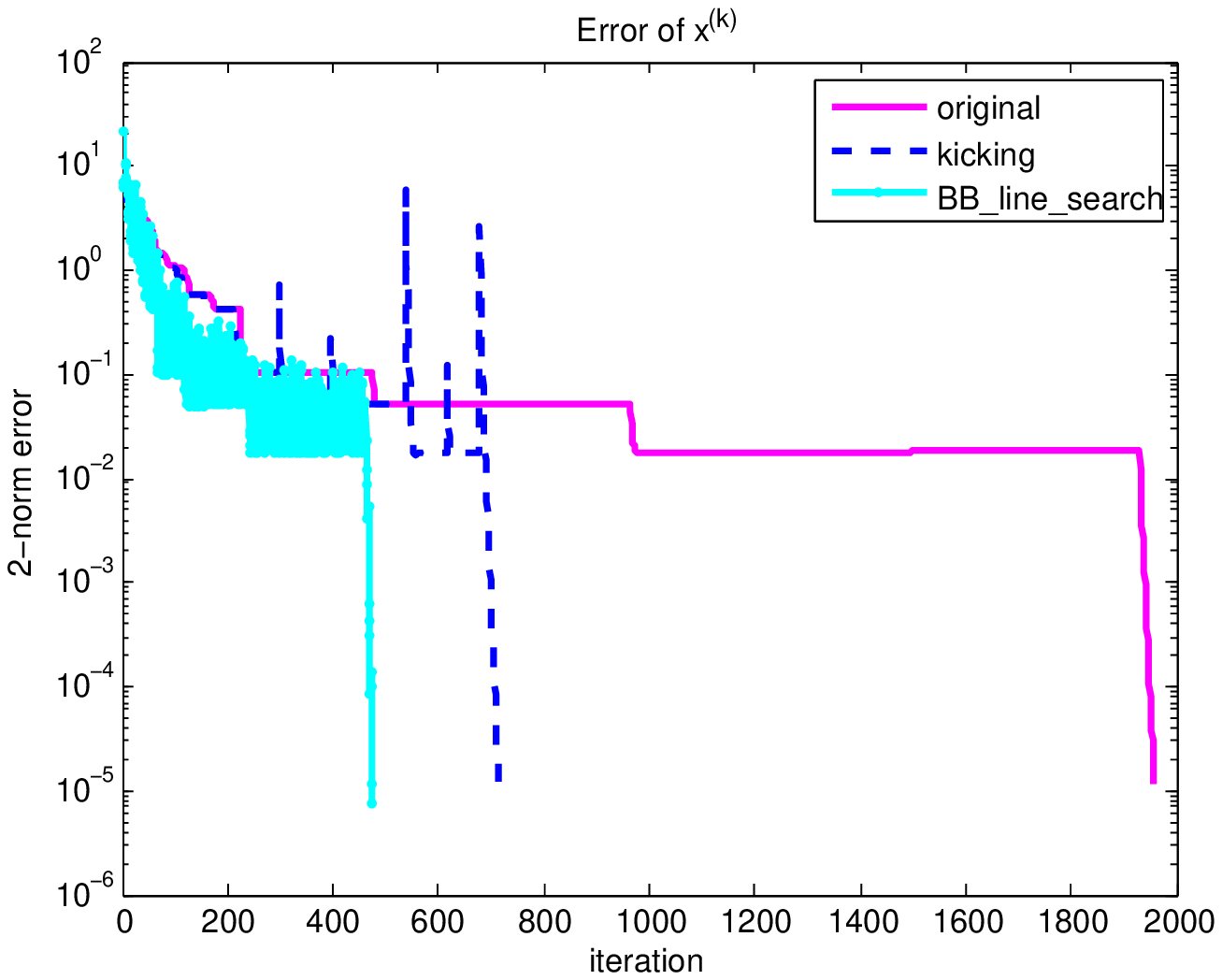}
\label{fig:fig1}
}
\subfigure[$\ell_2$ error of dual variable $\vy^{(k)}$]
{
\includegraphics[width=0.45\textwidth]{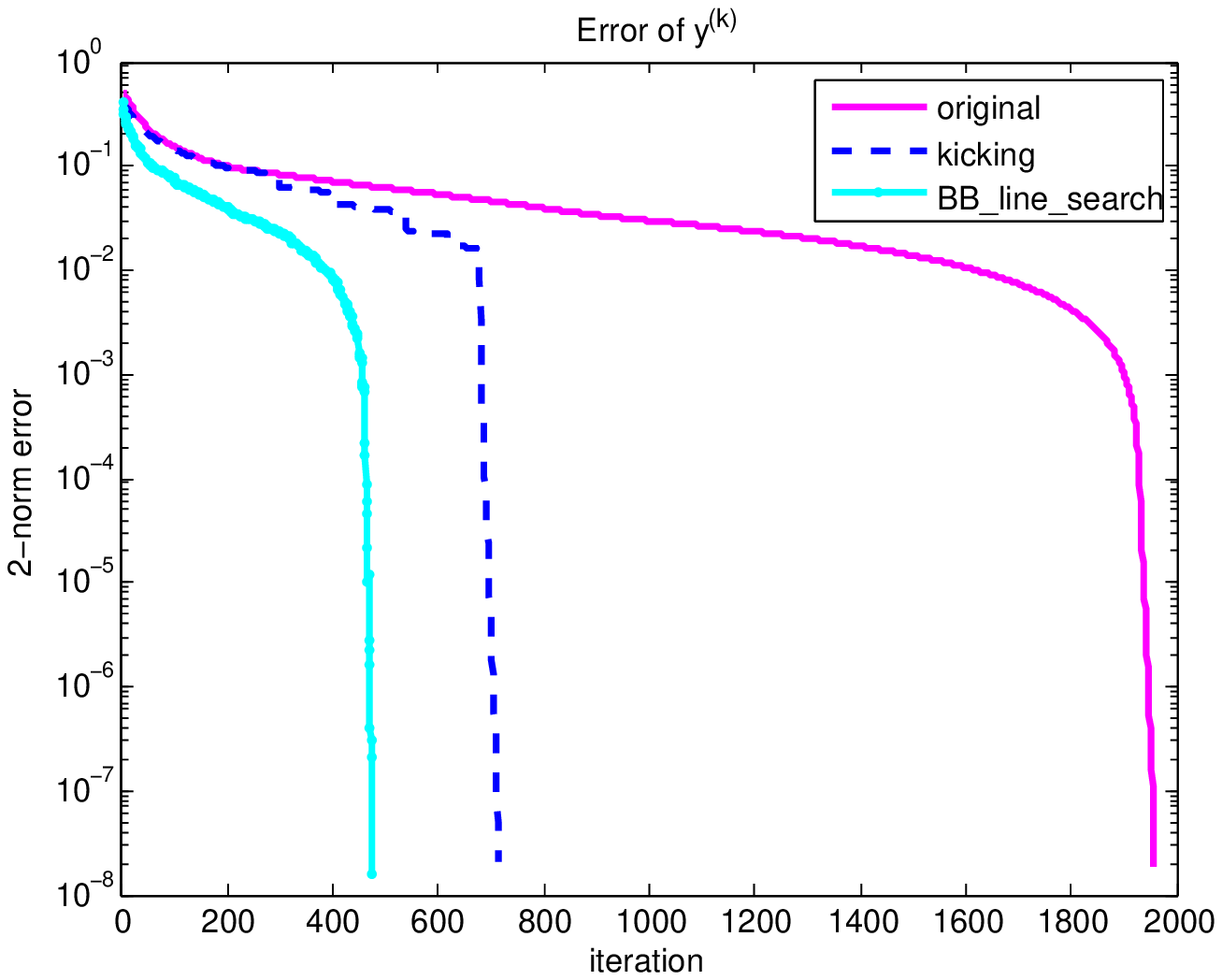}
\label{fig:fig2}
}
\caption{\label{fig:convg}Convergence of primal and dual variables of  three algorithms on \textbf{Gaussian sparse} $\vx^0$}
\end{figure}
\begin{figure}[ht]
\centering
\subfigure[$\ell_2$ error of primal variable $\vx^{(k)}$]
{
\includegraphics[width=0.45\textwidth]{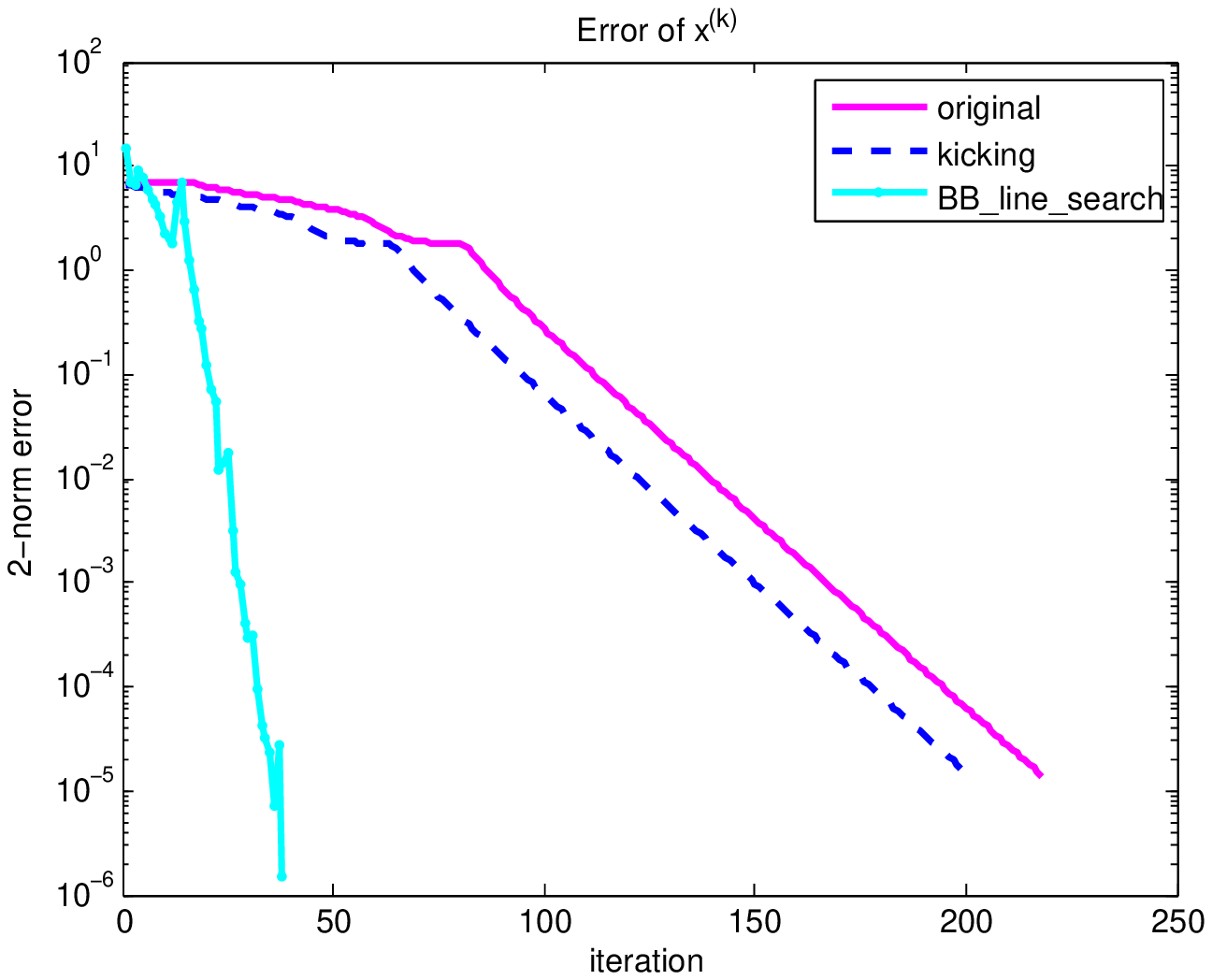}
\label{fig:fig1_flat}
}
\subfigure[$\ell_2$ error of dual variable $\vy^{(k)}$]
{
\includegraphics[width=0.45\textwidth]{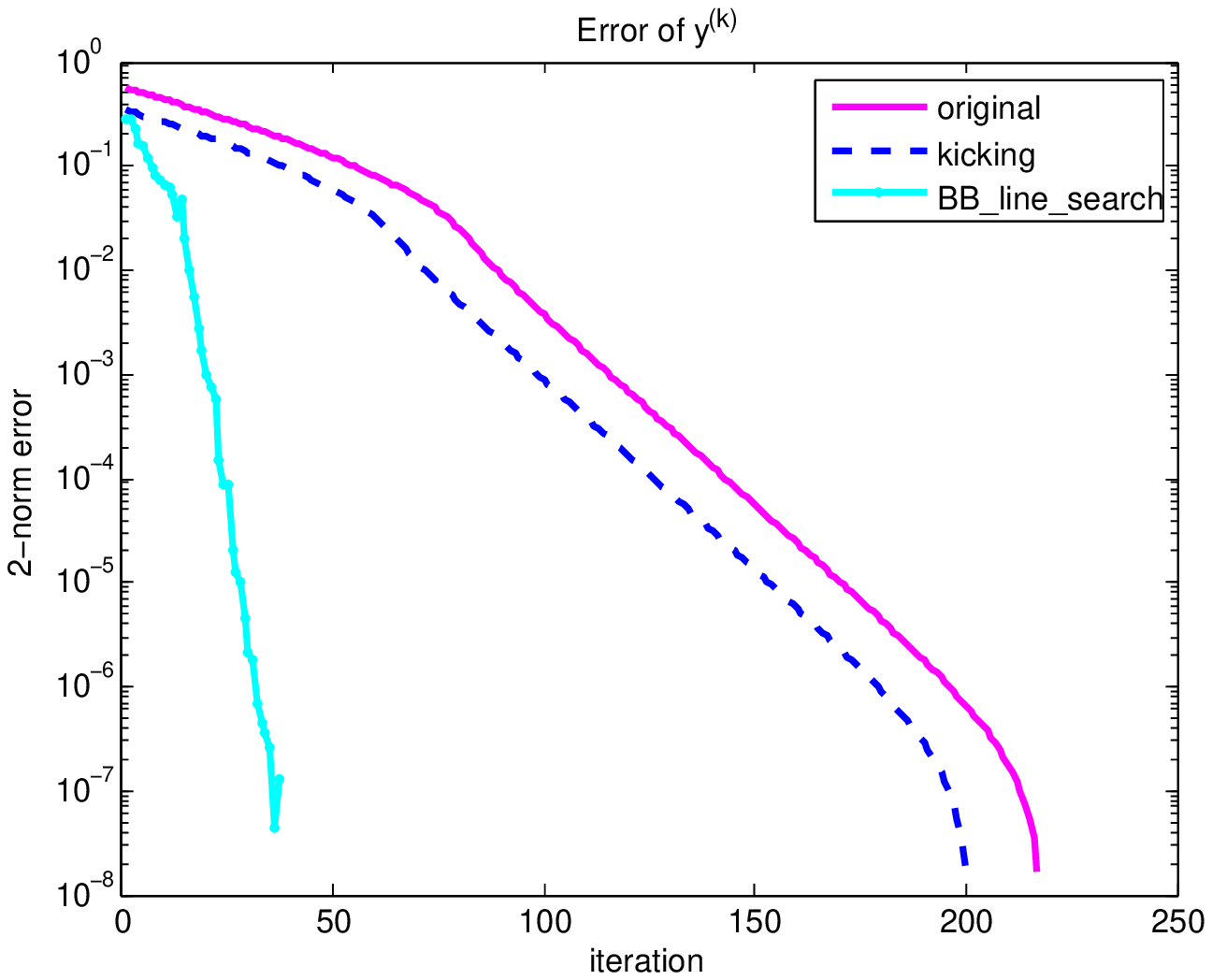}
\label{fig:fig2_flat}
}
\caption{\label{fig:convg_flat}Convergence of primal and dual variables of three  algorithms on \textbf{Bernoulli sparse} $\vx^0$}
\end{figure}
In both tests, the original version was the slowest. Besides the obvious speed differences, we observe that $\{\vx^{(k)}\}$ were not monotonic, there were sets of consecutive iterations in which $\vx^{(k)}$ did not change or fluctuated. 
Indeed, it is impossible to improve its R-linear convergence to Q-linear convergence. In addition, unlike the other two algorithms, \emph{BB-line-search} has non-monotonic $\{\vy^{(k)}\}$, which converges R-linearly instead of Q-linearly.

The convergence appears to have different stages. The early-middle stage has much slower convergence than the final stage.

Comparing the results of two tests, the convergence was faster on the Bernoulli sparse signal than the Gaussian sparse signal. Since the two tests used the same sensing matrix $\vA$ and the same sparsity, the main reason should be the dynamic range of the signals. A smaller dynamic range leads to faster convergence, which matches our theoretical result on the convergence rate. 

\section*{Appendix}
\begin{proof}[Proof of Theorem \ref{mainM0}]
We establish the theorem by showing that \eqref{nspMC} holds for any $\vH\in\Null(\cA)\setminus\{\vzero\}$.

Based on the SVD $ \vH=\sum_{i=1}^{m}\sigma_{i}(\vH) \vu_{i}\vv_{i}^\top$, where $\sigma_{i}(\vH)$ is the $i$-th largest singular value of $\vH$, we decompose $\vH= \vH_0+ \vH_1 + \vH_2+ \cdots$ where $\vH_0 = \sum_{i=1}^r \sigma_i(\vH) \vu_i\vv_i$, $\vH_1 = \sum_{i=r+1}^{2r} \sigma_i(\vH) \vu_i\vv_i$, $\vH_2 = \sum_{i=2r+1}^{3r} \sigma_i(\vH) \vu_i\vv_i$, \ldots. Following these definitions, condition \eqref{nspMC} can be equivalently written as
\beq\label{vh0s}
\|\vH_{0}\|_{*}<\|\sum_{i\ge 1}\vH_i\|_{*}.
\eeq
From $\vH\not=\vzero$ and the definition of $\vH_0$, we know that $\vH_0\not=\vzero$ and thus $\cA(\vH_0)\not=\vzero$ due to the RIP of $\cA$. From $\cA(\vH)=\vzero$ and $\cA(\vH_0)\not=\vzero$, it follows that $\cA(\sum_{i\ge 1}\vH_i)\not=\vzero$ and thus $\sum_{i\ge 1}\vH_i\not=\vzero$. Therefore, $\sum_{i\ge 1}\|\vH_i\|_*>0$, and we can define $t := \|\vH_{1}\|_{*}/(\sum_{i\ge 1}\|\vH_{i}\|_{*})>0$ and $\rho:=\|\vH_0\|_*/(\sum_{i\ge 1}\|\vH_{i}\|_{*})>0$.

Next, we present two inequalities without proofs (the interested reader can verify them following the proofs of Lemmas 2.3 and 2.4 in \cite{Mo-Li-11}):
\begin{subequations}\label{mtxlms}
\begin{align}
\frac{1-\delta_{2r}}{r}(\rho^{2}+t^{2})
\left(\sum_{i\ge 1}\|\vH_{i}\|_{*}\right)^{2}& \le \|\cA(\vH_0+\vH_1)\|_2^2,\\
\frac{t(1-t)
+\delta_{2r}(1-3t/4)^{2}}{r} \left(\sum_{i\ge 1}\|\vH_{i}\|_{*}\right)^{2} &\ge \|\cA(\sum_{i\ge 2}\vH_i)\|_2^2.
\end{align}
\end{subequations}
Since $\cA(\vH_0+\vH_1)+\cA\left(\sum_{i\ge 2}\vH_i\right)=\cA(\vH)=\vzero$, the two right-hand sides of \eqref{mtxlms} equal each other. Hence,
$$
\frac{1-\delta_{2r}}{r}(\rho^{2}+t^{2})
\left(\sum_{i\ge 1}\|\vH_{i}\|_{*}\right)^{2}\leq \frac{t(1-t)
+\delta_{2r}(1-3t/4)^{2}}{r} \left(\sum_{i\ge 1}\|\vH_{i}\|_{*}\right)^{2}
$$
and thus,
 $$
 \rho^2 \leq \frac{t(1-t)+\delta_{2r}(1-3t/4)^{2}-(1-\delta_{2r})t^{2}}{1-\delta_{2r}}.
 $$
 or
after a simple calculation of the maximum of $t\in [0, 1]$,
\beq\label{t2r}
\rho \leq\sqrt{\frac{4(1+5\delta_{2r}-4(\delta_{2r})^{2})}{(1-\delta_{2r})(32-25\delta_{2r})}}=:\theta_{2r}.
\eeq
If $\delta_{2r}<(77-\sqrt{1337})/82\approx0.4931$,
then $\theta_{2r}<1$ and thus $\rho<1$. By  definition, we get \eqref{vh0s} and \eqref{nspMC}. \end{proof}

\section*{Acknowledgements}
We thank Hui Zhang, who was visiting Rice from National U of Defense Technology, for his suggestions on the RIPless analysis, as well as Profs. Shiqian Ma and Qing Ling for valuable discussions. We also thank the anonymous referees for numerous suggestions and corrections that have helped improve this manuscript.

\bibliographystyle{plain}
\bibliography{Nonconvex,inverse,my_publications,newref}
\end{document}